\documentclass[12pt]{article}

\usepackage{url}
\usepackage{bbm}
\usepackage{mathtools}
\usepackage{amssymb}
\usepackage{amsthm}
\usepackage{empheq}
\usepackage{latexsym}
\usepackage{enumitem}
\usepackage{eurosym}
\usepackage{dsfont}
\usepackage{appendix}
\usepackage{color} 
\usepackage{hyperref}
\usepackage{frcursive}
\usepackage[utf8]{inputenc}
\usepackage[T1]{fontenc}
\usepackage{geometry}
\usepackage{multirow}
\usepackage{todonotes}
\usepackage{lmodern}
\usepackage{anyfontsize}
\usepackage{pgfplots}
\usepackage{stmaryrd}
\usepackage{natbib}

\setcitestyle{numbers,open={[},close={]}}

\pgfplotsset{compat=1.14}

\definecolor{red}{rgb}{0.7,0.15,0.15}
\definecolor{green}{rgb}{0,0.5,0}
\definecolor{blue}{rgb}{0,0,0.7}
\hypersetup{colorlinks, linkcolor={red},citecolor={green}, urlcolor={blue}}
			
\makeatletter \@addtoreset{equation}{section}

\newtheorem{theorem}{Theorem}[section]
\newtheorem{assumption}[theorem]{Assumption}
\newtheorem{corollary}[theorem]{Corollary}

\newtheorem{lemma}[theorem]{Lemma}

\newtheorem{definition}[theorem]{Definition}
\newtheorem{remark}[theorem]{Remark}

\def \F{\mathbb{F}}

\def \N{\mathbb{N}}

\def \P{\mathbb{P}}

\def \R{\mathbb{R}}

\def \Z{\mathbb{Z}}

\def\Ac{{\cal A}}

\def\Dc{{\cal D}}

\def\Fc{{\cal F}}

\def\Oc{{\cal O}}

\def\Qc{{\cal Q}}
\def\Rc{{\cal R}}

\def\Zc{{\cal Z}}

\title{Optimal make take fees in a multi market maker environment\footnote{This work benefits from the financial support of the Chaires Analytics and Models for Regulation, Financial Risk and Finance and Sustainable Development. Bastien Baldacci and Mathieu Rosenbaum gratefully acknowledge the financial support of the ERC Grant 679836 Staqamof. Dylan Possama\"i gratefully acknowledges the support of the ANR project PACMAN ANR-16-CE05-0027. The authors would like to thank Thibaut Mastrolia, Nizar Touzi, Steve Shreve and Ben Weber for pointing out some inaccuracies in an earlier version of the paper.}}
\author{Bastien {\sc Baldacci}\footnote{\'Ecole Polytechnique, CMAP, Route de Saclay, 91128, Palaiseau Cedex, France,  bastien.baldacci@polytechnique.edu.} \and 
Dylan {\sc Possama\"i} \footnote{ETH Z\"urich, Department of Mathematics, R\"amistrasse 101, 8092 Z\"urich, Switzerland, dylan.possamai@math.ethz.ch.} \and Mathieu {\sc Rosenbaum}\footnote{\'Ecole Polytechnique, CMAP, Route de Saclay, 91128, Palaiseau Cedex, France, mathieu.rosenbaum@polytechnique.edu}}

\begin{document}

\maketitle
\begin{abstract}
Following the recent literature on make take fees policies, we consider an exchange wishing to set a suitable contract with several market makers in order to improve trading quality on its platform. To do so, we use a principal-agent approach, where the agents (the market makers) optimise their quotes in a Nash equilibrium fashion, providing best response to the contract proposed by the principal (the exchange). This contract aims at attracting liquidity on the platform. This is because the wealth of the exchange depends on the arrival of market orders, which is driven by the spread of market makers. We compute the optimal contract in quasi explicit form and also derive the optimal spread policies for the market makers. Several new phenomena appears in this multi market maker setting. In particular we show that it is not necessarily optimal to have a large number of market makers in the presence of a contracting scheme.

\end{abstract}

\section{Introduction}\label{Section Introduction}

Optimal market making has been a topic of interest in mathematical finance since the seminal work \cite{avellaneda2008high}. Market makers are liquidity providers, who post limit orders on the bid and ask sides of the order book of an underlying asset, available on an exchange. They buy and sell simultaneously, earning the spread between their quotes and the mid-price, and have to dynamically manage their inventory, thus skew their quotes depending on their position. First simple market making problems are addressed in \cite{avellaneda2008high,ho1981optimal}, where the authors use a stochastic control approach. Later in \cite{gueant2013dealing}, an explicit solution is provided by imposing an inventory threshold for the market maker. A vast literature has emerged from these articles, and various extensions have been studied, see for example \cite{cartea2016closed,gueant2016financial}. All these models deal with the case of a market making activity with no maker taker fees policy from the exchange. 

\medskip
Due to the fragmentation of financial markets, exchanges  (Nasdaq, Euronext,\dots) are in competition and therefore need to find innovative ways to attract liquidity on their platforms. One of these ways is the use of a maker taker fees system: the exchange typically associates a fee rebate to executed limit orders, while charging a transaction cost for market orders. This enables it to subsidise liquidity provision and tax liquidity consumption. The problem of a relevant make take fees policy is therefore key for the quality of market liquidity and for the revenue of the corresponding exchange platform. In \cite{el2018optimal}, an optimal make take fees policy is derived, depending on the transaction flow generated by the market maker. 

\medskip
However, either with or without the intervention of an exchange, no optimal market making framework addresses the issue of several market makers competing with each other. Such consideration is of crucial importance for several reasons. The case of a single market maker means that he has no competitor, hence only needs to manage his inventory risk. However, for many financial assets, market making activity is provided by several market makers (typically three to ten), see for example \cite{mounjid2019asymptotic}. Therefore, having a multi agents model suits better to the vast majority of markets. A single agent model may overestimate the order flow that a market maker will process during the trading period.
Moreover, several important features of financial markets are linked to competition between market makers, such as spread formation and order book shape, see \cite{dayri2015large,glosten1985bid,huang2019glosten,madhavan1997security}. From the viewpoint of the exchange, considering several market makers can be relevant since it has in practice access to information related to the identity of the agents involved in each transaction. For those reasons, in the spirit of \cite{el2018optimal}, we extend the framework of make take fees problem to the case of an exchange (or of a regulator) wishing to attract liquidity on a market, with several market makers trading on a single underlying asset. 

\medskip
From a modelling point of view, our paper follows the same inspiration as the literature mentioned above. Here, we aim at studying the contracting problem of an exchange and several market makers who trade on a single underlying asset. We place ourselves in a principal-agent framework similar to \cite{holmstrom1987aggregation}. The principal wants to build a contract for the agents, enabling him to define an optimal make take fees policy, in order to maximise his revenue. Market takers send bid and ask orders, of constant volume equal to one, whereas market makers control their quotes on the asset (no notion of volume of limit orders is considered here). We base our contract on the transaction flow, as well as the asset price. Note that, as in \cite{el2018optimal}, the spread of the agents is not a contractible variable. It is indeed preferable to only consider variables involved in actual transactions. Furthermore, the very definition of individual spreads is ambiguous, as in practice market makers operates in a limit order book where they post several orders with different volumes. The market orders are executed by the market maker with the best quotes. However, other market participants receive a compensation depending on the distance between their quotes and the best bid or ask. By doing so, we aim at modelling queue position affecting the order book, see \cite{glosten1985bid,huang2019glosten,moallemi2016model}. As we assume constant volume, this represents the fact that an order may be splitted between agents depending on their position in the order book. Furthermore, this compensation reflects the importance for market makers to get tight enough quotes for commercial reasons, see Section \ref{Section market maker's problem} for details. Consequently, the PnL process of each market maker depend on the quotes of the others. Finally, in the single market maker framework of \cite{el2018optimal}, intensity of arrival orders are a decreasing function of the agent's spread. In our case, the intensity increases with total liquidity on the market. To do so we represent the whole liquidity of the order book using a weighted sum of the spreads depending on their distance with respect to an efficient price. 

\medskip
As in \cite{el2018optimal}, our problem is addressed by solving a Stackelberg game between the exchange and the market makers:

\begin{enumerate}[label=$(\roman*)$]
\item Following the approach in \cite{elie2016contracting}, each market maker computes his best--reaction function given spreads from other agents, which provides a Nash equilibrium.

\item Given this equilibrium, the exchange computes his optimal response (the optimal contract) by solving the appropriate Hamilton--Jacobi--Bellman equation.

\item This response is re--injected in the agents' optimal quotes, which give the optimal answer of both parties.
\end{enumerate}

\medskip
One can retrieve optimal quotes in a semi explicit form through a partial differential equation. Moreover, the optimal contract\footnote{As would be done in most cases in practice, we impose that every market maker receives the same contract. This means that it can not depend on individual risk aversion parameters.} is expressed as a sum of stochastic integrals with respect to market order and efficient price processes. We emphasize that such contracting scheme can be readily implemented and is easily interpretable, see Section \ref{subsection shape compensation}.

\medskip
An important finding is that an increase of the number of market makers does not necessarily decrease the average spread. It means that there is, for a given set of market parameters, an "optimal" number of market makers in the sense of PnL maximisation for the platform. It therefore provides, for an exchange, a framework to decide how many market participants they wish to attract in order to increase their profit. We also provide a simple formula to choose the "taker cost", in the same spirit as in \cite{el2018optimal}. This is discussed through Section \ref{Section Discussion}. We emphasize the fact that, to our knowledge, we provide the first paper using a framework \`a la Avellaneda and Sto\"ikov dealing with the multi market maker problem. One of our main contribution from this paper lies in the possibility to analyse the impact of adding market makers in a market on quantities of interest such as trading cost, total order flow, and PnL of the platform. We also see that adding a market participant with a higher risk-aversion parameter will decrease the average spread and conversely. Moreover, decreasing the taker cost when the number of market makers increase leads to a higher PnL for the platform up to a certain point (see Section \ref{Section Discussion}). 

\medskip
We organize the paper as follows: in Section \ref{Section The Model}, we introduce some preliminaries on stochastic calculus for the main objects of the model. Then, we present the way we model the multi market maker case, and the key differences with \cite{el2018optimal}. We also describe the market makers and the exchange's optimisation problem. In Section \ref{Section market maker's problem}, we give the best reaction functions of each market makers for a given contract and define the form of admissible contracts. In Section \ref{Section principal's problem}, we solve explicitly the exchange's problem and provide the form of the incentives given to each market maker. Finally, in Section \ref{Section comparaison}, we discuss the impact of the presence of several market makers on market liquidity and PnL of the platform. 

\medskip
{\bf Notations:} \noindent
Let $\N^\star$ be the set of all positive integers. For any $(\ell,c)\in \mathbb{N}^{\star} \times \mathbb{N}^{\star}$, $\mathcal{M}_{\ell,c}(\mathbb{R})$ will denote the space of $\ell\times c$ matrices with real entries. Elements of the matrix $M\in \mathcal{M}_{\ell,c}$ are denoted $(M^{i,j})_{1\leq i\leq \ell, 1\leq j\leq c}$ and the transpose of $M$ is denoted $M^{\top}$. We identify $\mathcal{M}_{\ell,1}$ with $\mathbb{R}^{\ell}$. When $\ell=c$, we let $\mathcal{M}_{\ell}(\mathbb{R}):=\mathcal{M}_{\ell,\ell}(\mathbb{R})$. For any $x\in \mathcal{M}_{\ell,c}(\mathbb{R})$, and for any $1\leq i\leq \ell$ and $1\leq j\leq c$, $x^{i,:}\in \mathcal{M}_{1,c}(\mathbb{R})$, and $x^{:,j}\in \mathbb{R}^{\ell}$ denote respectively the $i-$th row and the $j-$th column of $M$. Moreover, for any $x \in \mathcal{M}_{\ell,c}(\mathbb{R})$ and any $1\leq j\leq c$, $x^{:,-j}\in \mathcal{M}_{\ell,c-1}(\mathbb{R})$ denote the matrix $x$ without the $j-$th column. For any $x\in \mathcal{M}_{\ell,c}(\mathbb{R})$ and $y\in \mathbb{R}^{\ell}$, we also define for $i=1,\dots,\ell$, $y\otimes_{i}x \in \mathcal{M}_{\ell,c+1}(\mathbb{R})$ as the matrix whose first $i-1$ columns are equal to the first $i-1$ columns of $x$, such that for $j=i+1,\dots,c+1$, the $j-$th column is equal to the $(j-1)-$th column of $x$, and whose $i-$th column is equal to $y$. For any $x\in \R^\ell$, we also define $\underline x:=\min_{i\in\{1,\dots,\ell\}}x^i$. We also define $1_N$ the vector of $\mathbb{R}^N$ with every component equal to one. 

\medskip
Throughout the paper, we fix a constant $\delta_\infty>0$, which is assumed to be sufficiently large (how large exactly will be made specific later on).

\section{The model}\label{Section The Model}

\subsection{Framework}

\subsubsection{Canonical process}\label{Subsubsection Canonical Process}
As in \cite{el2018optimal}, the framework considered throughout this paper is inspired by the seminal works \cite{avellaneda2008high} and \cite{gueant2013dealing} where there is no exchange acting on the market. Let $N\geq 1$ be an integer representing the number of market makers in the market. We also define a positive constant $\delta_\infty$, which is assumed to be large enough, a statement that will be made precise later on. We consider a final horizon time $T>0$, and the space $\Omega=:\Omega_{c} \times \Omega_{d}^{2N}$, with $\Omega_{c}$ the set of continuous functions from $[0,T]$ into $\mathbb{R}$, and $\Omega_{d}$ the set of piecewise constant c\`adl\`ag functions from $[0,T]$ into $\mathbb{N}$. We consider $\Omega$ as a subspace of the Skorokhod space $\Dc([0,T],\R^{2N+1})$ of c\`adl\`ag functions from $[0,T]$ into $\R^{2N+1}$, and let $\Fc$ be the trace Borel $\sigma-$algebra on $\Omega$, where the topology is the one associated to the usual Skorokhod distance on $\Dc([0,T],\R^{2N+1})$. We let $(\mathcal{X}_{t})_{t\in [0,T]}:=\big(W_{t},N_{t}^{1,a},\dots,N_{t}^{N,a},N_{t}^{1,b},\dots,N_{t}^{N,b}\big)_{t\in [0,T]}$ be the canonical process on $\Omega$, that is to say
\begin{align*}
&W_{t}(\omega):=w(t), \; N_{t}^{i,a}(\omega):=n^{i,a}(t),\;  N_{t}^{i,b}(\omega):=n^{i,b}(t),\; \text{for all}\; i\in\{1,\dots,N\},
\end{align*}
with $\omega:=(w,n^{1,a},\dots,n^{N,a},n^{1,b},\dots,n^{N,b})\in \Omega.$ The following aggregated counting processes will also be useful
\begin{align*}
N^{j}:=\sum_{i=1}^{N}N^{i,j},\; j\in\{a,b\}.
\end{align*}
Finally, define for $i\in\{1,\dots,N\},j\in\{a,b\}$ the maps $\lambda^{i,j}:\R^{N}\times \mathbb{Z}^{N} \longrightarrow \R$ and $\lambda:\R^{N}\times \mathbb{Z}^{N} \longrightarrow \R$ by\footnote{See Section \ref{interpretations} for economical and mathematical interpretation of \eqref{Definition Intensity}.}
\begin{align}\label{Definition Intensity}
& \lambda^{i,j}(x,q):=A \exp\bigg(-\frac{k}{\sigma}\bigg(c+\varpi\sum_{i=1}^{N}x^{i}\mathbf{1}_{\{x^{i}=\underline x \}}+\sum_{i=1}^{N}\sum_{\ell=1}^{K}H_{\ell}x^{i}\mathbf{1}_{\{x^{i}\in K_{\ell}\}}\mathbf{1}_{\{x^{i}\neq \underline x\}}\bigg)\bigg)\frac{\mathbf{1}_{\{x^{i}=\underline x,q^{i}>-\phi(j)\overline{q}\}}}{\sum_{\ell=1}^{N}\mathbf{1}_{\{x^{\ell} = \underline x,q^{\ell}>-\phi(j)\overline{q}\}}},\\
\nonumber& \lambda^{j}(x,q):=\sum_{i=1}^{N}\lambda^{i,j}(x,q),
\end{align}
where $A$, $k$, $\sigma$, and $c$ are fixed positive constants, $K$ is a fixed positive integer, $\overline{q}\in \mathbb{N}$, $\varpi$ and $(H_{\ell})_{\ell=1,\dots,K}$ are real--valued and will be fixed later, and $\{K_{\ell}: \ell=1,\dots,K\}$ is an open covering of the interval $[0,\delta_{\infty}]$. Moreover
\begin{align*}
\phi(j):=
    \begin{cases}
        1, \text{ if } j=a,  \\
        -1, \text{ if } j=b.  
    \end{cases}
\end{align*}
\begin{remark}
The maps $\lambda^{i,j}$ are here to define the intensity of the point processes $N^{i,j}$, whereas $\lambda$ plays this role for the aggregated point processes. It is a generalisation of the exponential intensity used in {\rm\cite{el2018optimal,gueant2013dealing}}, in the sense that
\begin{itemize}
    \item all spreads are taken into account, weighted with respect to their value$;$.
    \item when $N=1$ and $\varpi=1$, we recover the intensity for the single market maker case. 
\end{itemize}
\end{remark}

\subsubsection{Admissible controls, inventory process and canonical probability measure}
We define the probability $\mathbb{P}^{0}$ on $(\Omega,\Fc)$ such that under $\P^0$, $W$, $N^{i,a}$, and $N^{i,b}$ are independent for all $i=1,\dots,N$, $W$ is a one--dimensional Brownian motion, $N^{i,a}$ and $N^{i,b}$ are Poisson processes with intensity $\lambda^{i,a}(0,0)$, and $\lambda^{i,b}(0,0)$ respectively.\footnote{As a direct consequence, $N^{a}$ and $N^{b}$ are Poisson processes with intensity $\lambda(0)$}\footnote{In other words, $\P^0$ is simply the product measure of the Wiener measure on $\Omega_c$ and the unique measure on $\Omega_d^{2N}$ that makes the canonical process there into an homogeneous Poisson process with the prescribed intensity.} We therefore endow the space $(\Omega,\mathcal{F})$ with the ($\P^0-$augmented) canonical filtration $\mathbb{F}:=(\mathcal{F}_{t})_{t\in[0,T]}:=(\mathcal{F}_{t}^{c}\otimes (\mathcal{F}_{t}^{d})^{\otimes d})_{t\in [0,T]}$ generated by $(\mathcal{X}_{t})_{t\in[0,T]}$. It is well--known that the filtration $\F$ satisfies the usual conditions and the Blumenthal $0-1$ law. All notions of measurability for processes, unless otherwise stated, should be understood as being associated to $\F$.

\medskip 
The market maker has a view on the efficient price (which should be understood as the mid--price) of the asset given by $(S_{t})_{t\in[0,T]}$, defined as
\begin{align}
S_{t}:=S_{0}+\sigma W_{t},\; t\in [0,T],
\label{Efficient Price}
\end{align}
where $S_0>0$ is the known initial value of the price, and $\sigma>0$ is its volatility.
\begin{remark}
The use of an arithmetic Brownian motion for the efficient price process is motivated by its simplicity. The price \eqref{Efficient Price} can reach negative values with non--negligible probability only on a sufficiently large time horizon $T$. For practical purposes, we choose $T< 1\text{ day}$ so that {\rm Equation \eqref{Efficient Price}} approximates accurately an efficient price driven by a geometric Brownian, which stays positive almost surely. 
\end{remark}

Next, we define the properties of the controlled processes of the agents. Based on their view on the efficient price \eqref{Efficient Price}, market makers offer bid and ask quotes on the underlying asset. Such prices are defined by
\[
P_{t}^{i,b}:=S_{t}-\delta_{t}^{i,b}, \; P_{t}^{i,a}:=S_{t}+\delta_{t}^{i,a}, \; t\in[0,T],\; i\in\{1,\dots,N\},
\]
where the superscript $b$ (resp. $a$) accounts for bid (respectively ask). The set of admissible controls for the market makers is therefore defined as
\begin{align}
\mathcal{A}:=\Big\{(\delta_{t})_{t\in [0,T]}=(\delta_{t}^{i,a},\delta_{t}^{i,b})_{t\in [0,T]}^{i=1,\dots,N}: \text{$\R^{2N}-$valued and predictable processes bounded by $\delta_{\infty}$} \Big\}.
\label{Admissible contracts market makers}
\end{align}
The predictability of the spreads reflects the fact that each agent chooses in advance his quotes. The boundedness assumption here is technical and simply helps us to define the associated probability measures. As the optimal contract we will derive later on leads naturally to bounded spreads, this is actually without loss of generality, provided the bounds are chosen large enough\footnote{See Lemma \ref{Lemma exchange Hamiltonian} for the prescribed value of $\delta_{\infty}$.}. 

\medskip
A market maker manages both the spreads and his inventory process. For the $i-$th agent, a filled bid order, represented by $N^{i,b}$, increase its inventory by one unit, and conversely for an ask order. It leads to the following definition of an inventory process of a market maker
\begin{align*}
Q_{t}^{i}:=N_{t}^{i,b}-N_{t}^{i,a},\; t\in [0,T],\; i\in\{1,\dots,N\}.
\end{align*}

\begin{remark}
Given the form of the intensities \eqref{Definition Intensity}, the $i-$th agent will see his inventory changing only if he quotes a spread such that $\delta^{i,b}=\underline{\delta}^{b}$, or $\delta^{i,a}=\underline{\delta}^{a}$. Such quotations are called the best bid, and best ask spread respectively. 

\medskip
The term $\overline{q}$ defined in \eqref{Definition Intensity} acts as a critical absolute inventory, which is the same for each agent. Assume the $i-$th market maker cross this threshold on the bid side, then $\lambda^{i,b}(\delta,q)=0$ and he can only receive ask orders to decrease his inventory below $\overline{q}$.  
\end{remark}

\subsubsection{Change of probability measure}
Given our technical assumptions, we introduce for any $\delta \in \mathcal{A}$ a new probability measure $\mathbb{P}^{\delta}$ on $(\Omega,\Fc)$ under which $S$ follows \eqref{Efficient Price} and 
\begin{align}\label{Definition Compensated Poisson Process}
\tilde{N}_{t}^{\delta,i,a}:=N_{t}^{i,a}-\int_{0}^{t}\lambda^{i,a}(\delta_{r}^{a},Q_{t})\mathrm{d}r,\; \tilde{N}_{t}^{\delta,i,b}:=N_{t}^{i,b}-\int_{0}^{t}\lambda^{i,b}(\delta_{r}^{b},Q_{t})\mathrm{d}r,  \; t\in[0,T],\; i\in\{1,\dots,N\},
\end{align}
are martingales. This probability measure is defined by the corresponding Dol\'eans-Dade exponential
\begin{align}\label{Definition change of variable}
&L_{t}^{\delta}:=\exp\Bigg(\sum_{i=1}^{N}\sum_{j\in\{a,b\}}\int_{0}^{t}\mathrm{log}\bigg(\frac{\lambda^{i,j}(\delta_{r}^{j},Q_{r})}{A}\bigg)\mathrm{d}N_{r}^{i,j}-\int_0^t\big(\lambda^{i,j}(\delta_{r}^{j},Q_{r})-A\big)\mathrm{d}r\Bigg),
\end{align}
where $Q:=(Q^{1},\dots,Q^{N})^\top$. By direct application of Itô's formula, and the uniform boundedness of $\delta^{a},$ and $\delta^{b}$, this local martingale satisfies the Novikov--type criterion given in  \cite{sokol2013optimal}, and thus is a martingale. 

\begin{remark}
By definition, the compensated aggregated point processes
\begin{align*}
\tilde{N}_{t}^{\delta,a}:=N_{t}^{a}-\int_{0}^{t}\lambda(\delta_{r}^{a})\mathrm{d}r,\; \tilde{N}_{t}^{\delta,b}:=N_{t}^{b}-\int_{0}^{t}\lambda(\delta_{r}^{b})\mathrm{d}r,\; t\in[0,T],
\end{align*}
are also martingales under $\mathbb{P}^{\delta}$.
\end{remark}
We can therefore define the Girsanov change of measure with $\frac{\mathrm{d}\mathbb{P}^{\delta}}{\mathrm{d}\mathbb{P}^{0}}=L_{T}^{\delta}$ (see for instance \cite[Theorem III.3.1]{jacod2003limit}). In particular, all the probability measures $\mathbb{P}^{\delta}$ indexed by $\delta \in \mathcal{A}$ are equivalent. The notation a.s., for almost surely, can be used without ambiguity. Throughout the paper, we write $\mathbb{E}^{\delta}_{t}$ for the conditional expectation with respect to $\mathcal{F}_{t}$ under the probability measure $\mathbb{P}^{\delta}$. 

\medskip
Hence, the arrival of ask (resp. bid) market orders for the $i-$th market maker is represented by the point process $(N_{t}^{i,a})_{t\in [0,T]}$ (resp. $(N_{t}^{i,b})_{t\in [0,T]}$) of intensity $(\lambda^{i,a}(\delta_{t}^{a})_{t\in [0,T]}$ (resp. $(\lambda^{i,b}(\delta_{t}^{a})_{t\in[0,T]}$) and the total arrival of ask (resp. bid) market orders is represented by the point process $(N_{t}^{a})_{t\in [0,T]}$ (resp. $(N_{t}^{b})_{t\in [0,T]}$) of intensity $(\lambda(\delta_{t}^{a})_{t\in [0,T]}$ (resp. $(\lambda(\delta_{t}^{b})_{t\in[0,T]}$).

\subsubsection{Interpretations}\label{interpretations}

First, we comment on the shapes of the intensities in \eqref{Definition Intensity}. The intensity of buy (resp. sell) market order arrivals depends on the extra cost of each trade paid by the market taker compared to the efficient price. This extra cost is the sum of the spread $\underline{\delta}^{b}$ (resp. $\underline{\delta}^{a}$) imposed by the market maker who is currently trading at the best bid (resp. best ask), and the transaction cost $c>0$ collected by the exchange. Moreover, following classical financial economics results, the average number of trades per unit of time is a decreasing function of the ratio between the spread and the volatility (see  \cite{dayri2015large}, \cite{madhavan1997security}, and \cite{wyart2008relation}). 

\medskip
The intensity of order arrivals depends also on the market liquidity, namely the spread quoted by all market participants. Hence, the spread of the $i-$th market maker is weighted by a constant $H_\ell,$  $\ell\in \{1,\dots,K\}$. Such constant will be chosen later on. For the moment, note that it is a decreasing function of $\ell$. Hence, a small spread corresponds to a high weight and conversely. Recall that in our model, we make the approximation that we can only have orders of size $1$. Hence, an increase of the intensity represents the fact that we can send bigger orders (many orders of size $1$ corresponds to one large order). 

\medskip
Such weights depends on the open covering $\big\{K_{\ell}:\ell\in\{1,\dots, K\}\big\}$, introduced in Subsection \ref{Subsubsection Canonical Process}. Several forms can be chosen (thinner intervals around the first Tick for instance), and throughout this paper we use the following definition
\begin{align}\label{Definition open covering}
K_{\ell}:=\big((\ell-1)\text{Tick}, \min(\ell\text{Tick},\delta_{\infty})\big),\; \ell\in\{1,\dots,K\}. 
\end{align}
Note that we can choose $K=1$, which leads to a unique zone $K_1=[0,\delta_\infty]$. In that case, the penalisation on the intensity of arrival orders is the same for all $\delta \neq \min_{i=1,\dots,N}\delta^i$. A larger number of intervals $K$ is used in order to have a penalisation increasing with respect to the value of the spread of the market makers. We will see that at the optimum, the covering has no impact on the PnL of both principal and agents, due to the form of both PnL process and intensity. 
\begin{remark}
{\rm Equation \eqref{Definition open covering}} indicates that the open covering is not a dynamic function of the vector of spreads. Hence, we add the indicator function $\mathbf{1}_{\{x_i\neq \underline x\}}$ to ensure that the coordinate $\underline x$ is not associated to a weight $\varpi + H_{\ell}$, for $\ell\in\{1,\dots, K\}$, but only to $\varpi$. In addition to this, fixing the open covering makes the model more tractable from a numerical point of view. Otherwise, at each time--step, the corresponding areas should be computed again. 
\end{remark}

\subsection{Market makers’ problem}

\subsubsection{PnL process of the agents}

First, we define the PnL process of the $i-$th market maker when the market makers play $\delta\in\Ac$ as
\begin{align}
 PL_{t}^{\delta,i}&:= \int_{0}^{t}\sum_{j\in\{a,b\}}\bigg(\delta_{s}^{i,j}\mathbf{1}_{\{\delta_{s}^{i,j}=\underline \delta_{s}^{j}\}}+\sum_{\ell=1}^{K}\omega_{\ell}\delta_{s}^{i,j}\mathbf{1}_{\{\delta_{s}^{i,j}\in K_{\ell}\}}\bigg)\mathrm{d}N_{s}^{j} +\int_{0}^{t}Q_{s}^{i}\mathrm{d}S_{s},\; t\in[0,T].
\label{PnL agent}
\end{align}

The first integral corresponds to the cash flow process, whereas the other represents the inventory risk process of the $i-$th agent, and $\omega_{\ell}\in(0,1),$ $\ell\in\{1,\dots,K\}$ are weights that decrease toward zero as $\ell$ increases. The market maker is remunerated with an increasing fraction of his quote when he is near from the best spread. Such remuneration decreases as the spread quoted moves away from the best spread. This is an incentive for the agents to quote a lower spread in general, and represents the fact that an order may be splitted between agents depending on their position in the order book.

\medskip
Such form of incentive is particularly well suited to certain markets where there are hundreds of market makers, and a selection of market participants has to be done by the platform. Such selection is required by some clients, who ask for a specific number of market maker to ensure competition on the exchange. Moreover, good position on external rankings attracts new clients on the platform. To do so, a selection criteria is the price quality provided by market makers.

\medskip
Moreover, we assume no partial execution for the agents. Hence, if two market makers are playing the best spread, they are both fully executed, in the sense that they receive the whole trade.  

\begin{remark}
Whereas the market maker is remunerated with a portion of his spread when trading in a certain area compared to the best bid--ask, it is not reflected in hid inventory process. Indeed, the terms $\omega_{\ell}\delta^{i,j}$ accounts only for the cash process: the market maker's inventory does not move if he does not quote at the best bid--ask.
\end{remark}

\subsubsection{Best reaction functions}

Equation \eqref{PnL agent} represents the PnL of an agent in the absence of contract from the exchange. Following the principal--agent approach, the exchange proposes a remuneration $\xi^{i}$ defined by an $\mathcal{F}_{T}-$measurable random variable to each market maker, in addition to their PnL process. These aim at creating an incentive to attract liquidity on the platform by reducing the market makers' spreads. We will prove a certain representation theorem for the set of admissible contracts.

\medskip
Given their contracts, the agents are facing a stochastic differential game, since the relative rankings of their spreads directly impact whether their orders are executed or not. Since we will later be looking for Nash equilibria for this game, we consider that the optimisation problem of the $i-$th market maker is a function of the spread vector $\delta^{-i}$ quoted by the $N-1$ other agents. Hence, each agent will maximise his PnL given actions of the other market makers to obtain his so--called best--reaction function. By denoting $\gamma_{i}>0$ the risk--aversion of the $i-$th market maker, and using a CARA utility function $U_i(x):=-\mathrm{e}^{-\gamma_i x}$, $x\in\R$, we are left with the following maximisation problem
\begin{align*}
&V_{\text{MM}}^{i}(\xi^{i},\delta^{-i}):=\!\!\!\!\!\sup_{\delta^{i}\in \mathcal{A}^{i}(\delta^{-i})}\!\!\!\mathbb{E}^{\mathbb{P}^{\delta\otimes_{i}\delta^{-i}}}\!\Bigg[U_i\Bigg(\xi^{i}\!+\!\!\!\!\sum_{j\in\{a,b\}}\!\!\int_{0}^{T}\!\delta_{t}^{i,j}\bigg(\mathbf{1}_{\{\delta_{t}^{i,j}=\underline{\delta_t^j\otimes_i\delta_t^{j,-i}}\}}\!+\!\sum_{\ell=1}^{K}\int_{0}^{T}\!\omega_{\ell}\mathbf{1}_{\{\delta_{t}^{i,j}\in K_{\ell}\}}\bigg)\mathrm{d}N_{t}^{j}\!+\!\!\int_{0}^{T}\!\!Q_{t}^{i}\mathrm{d}S_{t}\Bigg)\!\Bigg]\!,
\end{align*}
where $\mathcal{A}^{i}(\delta^{-i}):=\{ \delta:\delta\otimes_{i}\delta^{-i}\in \mathcal{A} \}$. 

\medskip
To ensure that this quantity is not degenerate, for $i\in\{1,\dots,N\}$, we assume that for all $\delta\in \mathcal{A}$
\begin{align}\label{Integrability market maker}
\mathbb{E}^{\delta}\big[\exp\big(-\gamma' \xi^i \big)\big]<+\infty, \text{ for some } \gamma'>\max(\gamma_{1},\dots,\gamma_{N}).
\end{align}
We call $\R^N-$valued $\Fc_T-$measurable random variables $\xi$ satisfying \eqref{Integrability market maker} contracts. The integrability condition ensures that the market maker's problem is well--defined (that is the sup remains finite). As the $N$ market makers play simultaneously, we are looking for a Nash equilibrium resulting of the interactions between agents. We now provide the appropriate definition of such an equilibrium

\subsubsection{Nash equilibrium}

A Nash equilibrium is a set of admissible controls such that each market maker has no interest in deviating from its current position, given a contract offered by the principal. It is formalised with the following definition.

\begin{definition} \label{Definition Nash equilibrium}
For a given contract $\xi$, a Nash equilibrium for the $N$ agents is a set of actions $\hat{\delta}(\xi)\in \mathcal{A}$ such that for all $i\in\{1,\dots,N\}$
\begin{align}\label{Nash Equilibrium First definition}
V_{\textup{MM}}^{i}(\xi^{i},\hat{\delta}^{-i}(\xi))\!=\!\mathbb{E}^{\mathbb{P}^{\hat\delta(\xi)}}\bigg[U_i\bigg(\xi^{i}\!\!+\!\!\!\sum_{j\in\{a,b\}}\!\int_{0}^{T}\hat\delta_{t}^{i,j}(\xi)\bigg(\mathbf{1}_{\{\hat\delta_{t}^{i,j}(\xi)=\underline {\hat\delta}_{t}^{j}(\xi)\}}\!\!+\!\sum_{\ell=1}^{K}\int_{0}^{T}\omega_{\ell}\mathbf{1}_{\{\hat\delta_{t}^{i,j}(\xi)\in K_{\ell}\}}\bigg)\mathrm{d}N_{t}^{j}\!+\!\!\int_{0}^{T}\!Q_{t}^{i}\mathrm{d}S_{t}\bigg) \bigg].
\end{align}
\end{definition}

We introduce for any contract $\xi$ the set $\text{NA}(\xi)$ of all associated Nash equilibria. This set is of particular importance as it will be imposed to be non--empty, to ensure the existence of, at least, one Nash equilibrium. We now turn to the exchange's contracting problem. 

\subsection{The exchange optimal contracting problem}

In our framework, the exchange is compensated by a fixed amount $c>0$ for each market order that occurs in the market. As in \cite{el2018optimal}, since we are anyway working on a short time interval, we take $c$ independent of the price of the asset. The goal of the exchange is to maximise the total number of aggregated market orders $N_{T}^{a}+N_{T}^{b}$
arriving during the time interval $[0,T]$. As the arrival intensities are only controlled by the market makers, the contract vector $\xi$ aim at increasing these intensities, which are decreasing functions of the spreads. Hence, the exchange will pay to each market maker this contract at time $T$ and the form of his PnL, using a CARA utility function, is given by
\begin{align*}
-\exp\Big(-\eta\big(c(N_{T}^{a}+N_{T}^{b})-\xi\cdot 1_N \big)\Big),
\end{align*}
where $\eta>0$ denote the risk aversion parameter of the principal.

\medskip
We now provide a suitable definition of the set of admissible contracts offered by the exchange. First, we need to ensure that the problem of the exchange does not degenerate. Hence, we assume that, for all $\delta\in \mathcal{A}$ and $i\in\{1,\dots,N\}$. 
\begin{align}
\mathbb{E}^{\delta}\big[\exp\big(\eta^\prime N\xi^i \big)\big]<+\infty, \text{ for some } \eta^\prime>\eta.
\label{Integrability principal}
\end{align}
Since $N^{a}$ and $N^{b}$ are point processes with bounded intensities, this condition, together with H\"older's inequality, ensure that the problem of the exchange is well--defined.\footnote{We will see in the verification Theorem \ref{Theorem Verification} that such a condition is required for a uniform integrability type argument} We also assume that the market makers only accept contracts $\xi^{i}$ such that their maximal utility $V^{i}_{\text{MM}}(\xi^{i},\hat{\delta}^{-i})$, taken at a Nash equilibrium $\hat{\delta}\in \text{NA}(\xi)$, is above a threshold value $R_{i}<0$. This value is known as the reservation utility of the $i-$th agent, and leads to the following definition.

\begin{definition}\label{Definition Admissible contract}
The set of admissible contracts $\mathcal{C}$ is defined as the set of $\R^N-$valued, $\Fc_T-$measurable random variables $\xi:=(\xi^{1},\dots,\xi^{N})^\top$, such that for all $i\in\{1,\dots,N\}$, \eqref{Integrability market maker} and \eqref{Integrability principal} hold, and the participation constraints of all agents are satisfied for at least one Nash equilibrium in ${\rm NA}(\xi)$ $($which is then automatically non--empty$)$.
\end{definition}

In the set of admissible contracts, the participation constraints of the agents are satisfied for at least one Nash equilibrium generated by $\xi$. As we anticipate that the participation constraints will be binding for any optimal contract $\xi$, the agents are indifferent between the possible different Nash equilibria generated by $\xi$. This means that we can use the same convention as the one used in the classical principal--agent literature, where the principal has enough bargaining power to impose to the agents which equilibrium he wants them to use. His optimisation problem is thus written as
\begin{align}
V_{0}^{E}:=\sup_{\xi \in \mathcal{C}}\sup_{\hat{\delta} \in \text{NA}(\xi)} \mathbb{E}^{\hat{\delta}(\xi)}\Big[-\exp\Big(-\eta\big(c(N_{T}^{a}+N_{T}^{b})-\xi\cdot 1_N \big)\Big)\Big].
\label{PnL principal}
\end{align} 
Now that we have properly defined the two problems of the Stackelberg game, we can move towards the resolution of the market maker's problem. Before solving this two--steps problem, we first sketch the approach we undertake.

\subsection{Stackelberg games in a nutshell}

Each market maker has an optimisation problem which depends on the control processes of the $N-1$ others agents, and on the contract given by the principal. Hence, solving the $i-$th agent's problem is done by searching the best reaction functions of each market makers, given a set of actions $\delta^{-i}$ of the other agents. Hence, the spreads quoted by every agent will both depend on the incentives given by the principal, and the spreads of opponents. 

\medskip
As stated before, the market makers fix their quotes simultaneously, so they must agree on an equilibrium between their reaction functions. To solve this problem, we will make use of an equivalent definition of a Nash equilibrium, given in \cite{elie2016contracting} and recalled in the next section. In particular, there is a direct link between the existence of a Nash equilibrium and a solution to a multidimensional system of BSDEs. 

\medskip
The key point is that, for a specific choice of weights $H_{\ell}$ in \eqref{Definition Intensity}, the existence and uniqueness of a Nash equilibrium is direct, because of two important facts. First, the use of indicator functions $\mathbf{1}_{\{\delta^{i,j}=\underline{\delta}^{j}\}}$, and $\mathbf{1}_{\{\delta^{i,j}\neq\underline{\delta}^{j}\}}$ acts as a decoupling effect on the agents' Hamiltonian. Second, such effect can be achieved only in the case of a restriction to a specific form of contracts. This restriction will be explained and commented in Section \ref{Section market maker's problem}. For the moment, note that it enables to compute explicitly a unique Nash equilibrium for the market maker's problem. 

\medskip
Given that admissible contracts generate at least one Nash equilibrium, the principal solve his optimisation problem by choosing the incentives given to each market maker, as a result of its associated HJB equation. This provides explicitly the optimal quotes of the agents, and solve the two steps Stackelberg game.

\section{Solving the market maker’s problem} \label{Section market maker's problem}

We start by solving the problem of the $i-$th market maker facing an arbitrary admissible contract proposed by the exchange. This section is mainly devoted to Theorem \ref{Theorem market maker}. First, we introduce a certain form of contracts proposed by the principal to the $i-$th agent. This $\mathcal{F}_{T}-$measurable random variable takes the form of the terminal condition to a specific BSDE, although we do not use this theory to solve the problem. We then prove that this is the only form of contracts that can be proposed to the market makers. Then, given this specific form, we derive the optimal response of each agents, other actions being fixed.

\subsection{Preliminaries}

For notational simplicity, let us define $\mathcal{R}:=\mathbb{R}^{N}\times \mathbb{R}^{N}\times \mathbb{R},\; \mathcal{B}_{\infty}:=[-\delta_{\infty},\delta_{\infty}].$
\begin{definition}\label{Definition Hamiltonian agent}
Fix some $i\in\{1,\dots,N\}$. For any $(d^{i},d^{-i},z^{i},q)\in \mathcal{B}_{\infty}^{2}\times\mathcal{B}_{\infty}^{2(N-1)}\times \mathcal{R}\times\mathbb{Z}^{N}$, where we have $z^{i}:=\big((z^{i,j,a})_{j=1,\dots,N},(z^{i,j,b})_{j=1,\dots,N},z^{S,i}\big)$, and $d^{i}:=(d^{i,a},d^{i,b})$, the Hamiltonian of the $i-$th agent is defined by
\begin{align}\label{Hamiltonian agent}
H^{i}(d^{-i}\!,z^{i},q)&:=\sup_{d^{i}\in \mathcal{B}_{\infty}^{2}}h^{i}(d^{i},d^{-i},z^{i},q),
\end{align}
where
\[
h^{i}(d^{i},d^{-i},z^{i},q):=\!\sum_{\ell=1}^{N}\sum_{j\in\{a,b\}}\!\!\gamma_{i}^{-1}\bigg(1-\exp\bigg(-\gamma_{i}\bigg(z^{i,\ell,j}+d^{i,j}\mathbf{1}_{\{d^{i,j}=\underline{d^j\otimes_i d^{j,-i}}\}}+\sum_{k=1}^{K}\omega_{k}d^{i,j}\mathbf{1}_{\{d^{i,j}\in K_{k}\}}\!\bigg)\!\bigg)\!\bigg)\lambda^{\ell,j}(d^{j},q).
\]
\end{definition}

As in every stochastic control problem, such quantity is of particular importance. It is naturally derived from an application of Itô's formula to $\mathrm{e}^{-\gamma_{i} Y_{t}}$, for $i=1,\dots,N$ where $Y$ is defined by \eqref{Process Yt}. Then, maximising this quantity will give the optimal spreads quoted by the $i-$th market maker, given the spreads quoted by the other agents.

\medskip
The form of the maps $(h^i)_{i=1,\dots,N}$ requires that a maximiser in \eqref{Hamiltonian agent} is a function of $d^{-i},$ $z^i$, and $q$. This suggests a proper definition of a fixed point of the Hamiltonian vector $\big(H^{i}(d^{-i},z^{i},q)\big)_{i=1,\dots,N}$.

\begin{definition}
For every $(z,q)\in \mathcal{R}^{N}\times \mathbb{Z}^{N}$ a fixed point of the Hamiltonian is defined by a matrix $\delta^{\star}(z,q)\in \mathcal{M}_{N,2}(\mathbb{R})$ such that for any $1\leq i \leq N$
\begin{align}\label{Definition Fixed Point}
\delta^{\star i}(z,q)\in \underset{\delta^{i}\in\mathcal{B}_{\infty}^{2}}{\textup{argmax }}h^{i}(\delta^{i},\delta^{\star-i},z^{i},q).
\end{align}
For every $(z,q)\in \mathcal{R}^{N}\times \mathbb{Z}^{N}$, we denote by $\mathcal{O}(z,q)$ the set of all fixed points.
\end{definition}
We need the following standing technical assumption.
\begin{assumption}\label{Assumption fixed point}
There exists at least one Borel--measurable map $\delta^{\star}:\Rc^N\times\Z^N\longrightarrow \mathcal{M}_{N,2}(\mathbb{R})$ such that for every $(z,q)\in \mathcal{R}^{N}\times \mathbb{Z}^{N}$, $\delta^{\star}(z,q)\in \mathcal{O}(z,q)$. The corresponding set of maps is denoted by $\Oc$.
\end{assumption}

\begin{remark}
We will see that in the specific case where $z^{i,j,a}=z^{i,a}$, and $z^{i,j,b}=z^{i,b}$ for all $(i,j)\in\{1,\dots,N\}^2$, there exists a unique fixed point in $\Oc(z,q)$ for any $(z,q)\in \mathcal{R}^{N}\times \mathbb{Z}^{N}$. This specification is used in {\rm Corollary \ref{Corollary best response agent}}, where the Nash equilibrium is provided explicitly.
\end{remark}

We now define a family of processes which represents the form of contract given to the agents. 

\begin{definition}\label{Definition Process Yt}
Given $y_{0}\in \mathbb{R^N}$, and $\Rc-$valued predictable process $Z^{i}:=(Z^{i,j,a},Z^{i,j,b},Z^{S,i})_{j=1,\dots,N}$, for $i\in\{1,\dots,N\}$, we introduce the family of $\R^N-$valued processes $(Y^{y_{0},Z,\hat{\delta}})_{\hat\delta\in\Oc}$ indexed by fixed point maps $\hat\delta\in\Oc$, whose $i-$th coordinate is given by, for $i\in\{1,\dots,N\}$ and $t\in[0,T]$
\begin{align}\label{Process Yt}
Y_{t}^{i,y_{0},Z,\hat{\delta}}\!:=\!y_{0}^{i}\!+\!\!\sum_{j=1}^{N}\int_{0}^{t}Z_{r}^{i,j,a}\mathrm{d}N_{r}^{j,a}\!+\!Z_{r}^{i,j,b}\mathrm{d}N_{r}^{j,b}\!+\!Z_{r}^{S,i}\mathrm{d}S_{r}\!+\!\bigg(\frac{\gamma_{i}\sigma^{2}}{2}(Z_{r}^{S,i}\!+\!Q_{r}^{i})^{2}\!-\!H^{i}\big(\hat{\delta}^{-i}(Z_r,Q_r),Z_{r}^{i},Q_{r}\big)\bigg)\mathrm{d}r.
\end{align}
We say that $Z:=(Z^i)_{i=1,\dots,N}$ belongs to the set $\Zc$, if $Y_{T}^{y_{0},Z,\hat{\delta}}$ satisfies \eqref{Integrability market maker}, \eqref{Integrability principal} and for all $\delta\in\mathcal{A}$,
\begin{align*}
   \mathbb{E}^\delta \Big[\sup_{t\in[0,T]}\exp\big(-\gamma_i^{'} Y_{t}^{i,y_{0},Z,\hat{\delta}}\big)\Big] <+\infty.
\end{align*}
\end{definition}
This condition ensures that the market maker's problem is not degenerated given this specific form of contract. Moreover, given the integrability conditions on the coefficients, the processes $\big(Y^{y_{0},Z,\hat{\delta}}\big)_{\hat{\delta}\in \mathcal{O}}$ are well defined and $\big(\mathrm{e}^{-\gamma_{i} Y^{i}}\big)$ is a uniformly integrable process under $\mathbb{P}^{\delta}$, for every $\delta\in \mathcal{A}$, and $i\in\{1,\dots,N\}$.\footnote{Such condition is used to provide explicitly the best response of the agents in Corollary \ref{Corollary best response agent}.} To link an admissible vector contract $\xi\in \mathcal{C}$ to the processes defined in \eqref{Process Yt}, we define the following set.

\begin{definition}
We define $\Xi$ as the set of random variables $Y_T^{y_{0},Z,\hat{\delta}}$ where $(y_0,Z,\hat\delta)$ ranges in $\R^N\times\Zc\times\mathcal{O}$,  and such that $\mathrm{e}^{-\gamma_i y_0^i}\geq R_i$ for any $i\in\{1,\dots,N\}$.\footnote{Theorem \ref{Theorem market maker} proves that such contract generates at least one equilibrium.}
\end{definition}

Since by definition all bounded predictable processes are contained in $\mathcal{Z}$, it is clearly nonempty. 

\medskip
To prove equality of these sets, we are  reduced  to  the  problem  of  representing  any  contract $\xi^{i}$ as $Y_{T}^{i,y_{0},Z,\hat{\delta}}$ for some $(y_{0},Z)\in \mathbb{R}^N\times\mathcal{Z}$ and some $\hat\delta \in \mathcal{O}$. Following the approach of \cite{sannikov2008continuous}, we derive a dynamic programming principle for the utility function of the market maker, and then prove the equality of the sets by identification of the coefficients. 

\subsection{Contract representation}
The following theorem provides solution to the market maker's problem, and a complete characterisation of the set of admissible contracts.

\begin{theorem}\label{Theorem market maker}
Any contract vector $\xi=Y_{T}^{y_{0},Z,\hat{\delta}}$ with $(y_{0},Z,\hat\delta) \in \mathbb{R}^{N}\times\mathcal{Z}\times\mathcal{O}$ leads to a unique Nash equilibrium for the agents, given by $\big(\hat\delta(Z_t,Q_t)\big)_{t\in[0,T]}$.

\medskip
Conversely, any admissible contract $\xi \in \mathcal{C}$ is of the form $\xi=Y_{T}^{y_{0},Z,\hat{\delta}}$ for some $(y_{0},Z) \in \mathbb{R}^{N}\times\mathcal{Z}$ and a certain $\hat\delta \in \mathcal{O}$. 
\end{theorem}
In the next corollary, we restrict ourselves to a subset of admissible contracts under which each agent earns at least his reservation utility, and where we can derive explicitly their best--response. For such purpose, we introduce the following set
\begin{align*}
\Xi^\prime \!:=\! \Big\{Y_{T}^{y_{0},Z,\hat{\delta}}\!\!:\!(\hat\delta,y_0,Z)\!\in\!\mathcal{O}\!\times\!\R^N\!\!\times\!\Zc,\;\! \text{\rm s.t. for all } (i,j,k)\!\in\!\{1,\dots,N\}^2\!\!\times\!\{a,b\},\;\! \mathrm{e}^{-\gamma_i y_0^i}\!\geq \!R_i,\! \; Z^{i,j,k}\!=:\!Z^{k}\Big\}.
\end{align*}
The main interest of the subset $\Xi^\prime$ is the following result.
\begin{lemma}\label{Lemma unicity fixed point Hamiltonian}
Assume that, for $(z,q)\in\mathcal{R}^N \times \mathbb{Z}^N$, we have $z^{i,\ell,j}=z^{j}$ for all $(i,\ell)\in \{1,\dots,N\}^2$ and $j\in\{a,b\}$.
We define 
\begin{align*}
&\Gamma^{i,j}(z):=-z^{j}+\frac{1}{\gamma_{i}}\mathrm{log}\bigg(1+\frac{\sigma\gamma_{i}}{k\varpi}\bigg),\; z\in\Rc^N.
\end{align*}
We also introduce the function $\Delta:\mathcal{R}^{N}\times\mathbb{Z}^N \longrightarrow \mathcal{M}_{N,2}(\mathbb{R})$ defined by, for $i\in\{1,\dots,N\}$, $j\in\{a,b\}$, $(z,q)\in\mathcal{R}^{N}\times\mathbb{Z}^N$
\begin{align}\label{Optimal response Corollary}
\Delta^{i,j}(z,q) := 
    \begin{cases}
   \displaystyle     (-\delta_{\infty})\vee \Gamma^{i,j}(z)\wedge \delta_{\infty},\; \mbox{\rm if}\; \Gamma^{i,j}(z) < \Gamma^{\ell,j}(z),\; -\overline{q}<q<\overline{q}, \;\mbox{\rm for all } \ell\neq i,  \\[0.8em]
    \displaystyle    (-\delta_{\infty})\vee \frac{1}{\omega_{\ell}}\Gamma^{i,j}(z)\wedge \delta_{\infty},\; \mbox{\rm if}\; \frac{1}{\omega_{\ell}}\Gamma^{i,j}(z)\in K_{\ell},\; -\overline{q}<q<\overline{q},\; \mbox{\rm for}\; \ell\in\{1,\dots,K\},\\
    \displaystyle 0,\; \text{\rm otherwise}.
    \end{cases}
\end{align}
Then, $\mathcal{O}(z,q)$ is reduced to the singleton $\big\{\Delta(z,q)\big\}$.
\end{lemma}
The proof is reported in the appendix, and follows from standard computations on the Hamiltonian \eqref{Definition Hamiltonian agent}. In particular, it leads to existence and uniqueness of the maximiser of \eqref{Definition Hamiltonian agent}. We can now conclude with the announced corollary. 
\begin{corollary}\label{Corollary best response agent}
For any admissible contract $Y_{T}^{y_{0},Z,\hat{\delta}}\in\Xi^\prime$ offered by the principal, there exist a unique Nash equilibrium, given by by $\big(\Delta(Z_t,Q_t)\big)_{t\in[0,T]}$, where the map $\Delta$ is defined in \eqref{Optimal response Corollary}. 
\end{corollary}
This result states that, at the optimum, the utility function of each market maker corresponds to its reservation utility, that is to say the quantity such that the $N$ agents accept their contract. Moreover, it enables us to characterise explicitly a unique Nash equilibrium for the market maker's problem. We end the section with some comments on the shape of admissible contracts \eqref{Definition Process Yt}. 

\subsection{On the shape of compensation proposed and contractible variables.}\label{subsection shape compensation}
In this section, we would like to highlight some interpretation on the classes $\Xi$ and $\Xi^\prime$ of "smooth" contracts $\xi=(\xi^1,\dots,\xi^N)^\top$ controlled by $(y_0,Z,\hat\delta)\in \R^N\times\Zc\times\mathcal{O}$ and having the form
\[\xi^{i}=y_{0}^{i}+\sum_{j=1}^{N}\int_{0}^{T}Z_{r}^{i,j,a}\mathrm{d}N_{r}^{j,a}+Z_{r}^{i,j,b}\mathrm{d}N_{r}^{j,b}+Z_{r}^{S,i}\mathrm{d}S_{r}+\bigg(\frac{1}{2}\gamma_{i}\sigma^{2}(Z_{r}^{S,i}+Q_{r}^{i})^{2}-H^{i}\big(\hat{\delta}^{-i}(Z_r,Q_r),Z_{r}^{i},Q_{r}\big)\bigg)\mathrm{d}r.
 \]
Note that the contracts are indexed on the number of transactions and on the efficient price of the asset. The spreads of the market makers are observed by the platform but are not contractible variables, since a contract depending on them would be unrealistic in practice, see \cite{el2018optimal}. Mathematically speaking, allowing the contract to depend on the spreads would correspond to a first best problem. For sake of completeness, we compute the solution of the first best problem in Appendix \ref{Section principal First Best} and show that it differs from the one considered here.
 \begin{enumerate}[leftmargin=*]
 \item[$\bullet$] The compensation $y_{0}^{i}$ is calibrated by the exchange to ensure the reservation utility constraint with level $R_i$ of the $i-$th market maker, we refer to Section \ref{Section principal's problem} for more details on it.
 \item[$\bullet$] The term $\int_0^T Z^{S,i}_r \mathrm{d}S_r$ is the compensation given to the market maker with respect to the efficient price.
 \item[$\bullet$] The terms $\int_0^T Z_r^{i,j,a} \mathrm{d}N_r^{j,a},$ and $\int_0^T Z_r^{i,j,b} \mathrm{d}N_r^{j,b}$ are the compensation of the $i-$th market maker with respect to the number of trades made on the ask side or bid side by the $j$-th market maker.
 \item[$\bullet$] The term $\int_0^T H^{i}\big(\hat{\delta}^{-i}(Z_r,Q_r),Z_{r}^{i},Q_{r}\big)\mathrm{d}r$ is the certain gain of the $i-$th agent induced by his maximisation problem. The principal anticipates that the agent will earn money coming from his maximisation strategy. Hence, he deducts such corresponding amount to the salary of the agent. This justifies why this term appears with a minus sign in the compensation $\xi^i$.
 \item[$\bullet$] The term $\int_0^T\frac{1}{2}\gamma_i \sigma^{2}(Z_{r}^{S,i}+Q^{i}_{r})^{2}\mathrm{d}r$ is the compensation\footnote{It corresponds to the quadratic variation of the sum of the incentive indexed on $S$ and the inventory process of the $i$-th market maker, integrated against $S$ and weighted by its risk aversion.} to balance the risk aversion of the agent with respect to the efficient price and his inventory.
 \end{enumerate}
From a representation viewpoint, the subset $\Xi^{\prime}$ means that we index the contract of the $i-$th market maker only on the aggregated order processes $N^{a}$ and $N^{b}$, and the efficient price $S$. Moreover, we restrict ourselves to a subset of the admissible contracts where the incentives with respect to the bid and ask arrival orders are equal for every agent. Hence, we do not discriminate a priori one market maker compared to another. However, the discrimination is done in the market risk part, namely $\int_{0}^{T}Z_{t}^{S,i}\mathrm{d}S_{t},$ for $i\in\{1,\dots,N\}$. This assumption is in force until the end of the paper (except in the appendix). Practically, the incentives for the $i-$th agent are only functions of his own inventory process, the aggregated order flow, and the efficient price. It appears reasonable from a practical point of view, as it means that the platform does not need to monitor cross incentives $Z^{i,j,a}$ or $Z^{i,j,b}$ for $j\neq i$, which is hard to do in practice. In addition to this, the exchange give the same incentives to the agents on the part driven by the market orders sent by market takers, but can discriminate with respect to the risk aversion parameters on the part driven by the efficient price around which market makers adjust their quotes.\\

From the technical point of view, this simplification enables to obtain an explicit formula for the fixed points of the Hamiltonian, which is not the case in the general framework. We will also see in the next section that this restriction drastically simplify computations to derive explicitly the optimal incentives that the principal provides to each market makers.
 
\section{Solving the principal's problem} \label{Section principal's problem}
Denote for all $i\in\{1,\dots,N\}$, $\hat{y}_{0}^{i}:=-\frac{1}{\gamma_{i}}\mathrm{log}(-R_{i})$. By Theorem \ref{Theorem market maker} and Corollary \ref{Corollary best response agent}, the exchange problem \eqref{PnL principal}, when restricted to contracts in $\Xi^\prime$, reduces to the control problem
\begin{align}
&\widetilde{V}_{0}^{E}:=\sup_{y_{0}\geq \hat{y}_{0}}\sup_{Z\in \mathcal{Z}}\mathbb{E}^{\Delta(Z,Q)}\Big[-\exp\Big(-\eta\big(c(N_{T}^{a}+N_{T}^{b})-Y_{T}^{y_{0},Z,\Delta(Z,Q)}\cdot 1_N \big)\Big)\Big].
\label{Reduced exchange problem}
\end{align}
Corollary \ref{Corollary best response agent} provides the best responses of the agents as a function of the control process $Z\in \mathcal{Z}$ of the principal.\footnote{In this section $\hat\delta$ is the unique Nash equilibrium of the agent's problem coming from Corollary \ref{Corollary best response agent}.} Given such response, the exchange solve \eqref{PnL principal}, with $\xi\in \Xi^{'}$, in two steps
\begin{itemize}
    \item Due to the form of utility function, the optimisation with respect to $y_0$ ensures the reservation utility constraint of the agents is satisfied.
    \item Optimisation with respect to $Z\in \mathcal{Z}$ is done by solving a classical Hamilton--Jacobi--Bellman equation associated to the reformulated control problem. 
\end{itemize}

The section ends with a verification argument to ensure that the value function coincides with \eqref{PnL principal}, and some comments on switching policy between market makers and physical interpretation of the results.

\subsection{Saturation of utility constraint}
Note that the market makers' optimal response in Corollary \ref{Corollary best response agent} does not depend on $y_{0}$. The exponential linear framework for the PnL of the principal enables to state directly that this objective function is clearly decreasing in all coordinates of $y_{0}$, implying that the maximisation under the participation constraint is achieved at $\hat{y}_{0}$
\begin{align*}
\widetilde{V}_{0}^{E}=\mathrm{e}^{\eta \hat{y}_{0}\cdot\mathbf{1}_N}\sup_{Z\in\mathcal{Z}}\mathbb{E}^{\Delta(Z,Q)}\bigg[-\exp\Big(-\eta\big(c(N_{T}^{a}+N_{T}^{b})-Y_{T}^{y_{0},Z,\Delta(Z,Q)}\cdot 1_N \big)\Big)\bigg].
\end{align*}
Hence, we are left with a maximisation problem with respect to $Z\in \mathcal{Z}$, which is nothing else than a standard stochastic control problem with the state variables $Q$, $N^a$, $N^b$ and $Y^{y_0,Z,\Delta(Z,Q)}$. 

\subsection{The HJB equation for the reduced exchange problem}

We study in this section the HJB equation corresponding to the stochastic control problem 
\begin{align}\label{Definition reduced exchange control problem}
v_{0}^{E}:=\sup_{Z\in \mathcal{Z}}\mathbb{E}^{\Delta(Z,Q)}\Big[-\exp\Big(-\eta\big(c(N_{T}^{a}+N_{T}^{b})-Y_{T}^{y_{0},Z,\Delta(Z,Q)}\cdot 1_N\big)\Big)\Big].
\end{align}
For the sake of simplicity, we define for any map $v:[0,T]\times\mathbb{Z}^N\longrightarrow (-\infty,0)$, any $x\in \mathbb{R}$, any $i\in\{1,\dots,N\}$, and any $(t,q)\in[0,T]\times\Z^N$
\begin{align*}
& v(t,q\oplus_i x):=v(t,q^{1},\dots,q^{i-1},q^{i}\!+\!x,q^{i+1},\dots,q^{N}),\; v(t,q\ominus_i x):=v(t,q^{1},\dots,q^{i-1},q^{i}\!-\!x,q^{i+1},\dots,q^{N}).
\end{align*}
We also define the maps $\mathcal{V}^{+}(t,q):=\big(v(t,q\oplus_i 1)\big)_{i=1,\dots,N}$, and $\mathcal{V}^{-}(t,q):=\big(v(t,q\ominus_i 1)\big)_{i=1,\dots,N}$, for $(t,q)\in[0,T]\times\Z^N$, as well as the set $\mathcal{Q}:=\{-\overline{q},\dots,\overline{q}\}$. The HJB equation associated to \eqref{Definition reduced exchange control problem} is
\begin{align}\label{HJB equation}
\begin{cases}
     \displaystyle   \partial_{t}v(t,q) + \mathcal{H}\big(q,\mathcal{V}^{+}(t,q),\mathcal{V}^{-}(t,q),v(t,q)\big)=0, \; (t,q)\in[0,T)\times\mathcal{Q}^{N},  \\[0.5em]
    \displaystyle    v(T,q)=-1,\; q\in\mathcal{Q}^{N},
\end{cases}
\end{align}
where $\mathcal{H}\big(q,p,m,v\big):=\mathcal{H}^{S}\big(q,v\big)+\mathcal{H}^{b}\big(q,p,v\big)+\mathcal{H}^{a}\big(q,m,v\big),$
with, for any $(q,p,\ell)\in\Qc^N\times\R^N\times\{a,b\}$
\begin{align*}
& \mathcal{H}^{S}\big(q,v\big)= \sup_{z^{S}\in\mathbb{R}^N}v\bigg(\sum_{i=1}^{N}\frac{\eta}{2}\sigma^{2}\gamma_{i}\big(z^{S,i}+q^{i}\big)^{2}+\frac{\eta^{2}\sigma^{2}}{2}\|z^S\|^2\bigg),  \\
& \mathcal{H}^{\ell}\big(q,p,v\big)=\sup_{z^{\ell}\in\mathbb{R}}\sum_{i=1}^{N}\lambda^{i,\ell}\big(\Delta^{:,\ell}(z,q),q\big)\Big(\mathrm{e}^{\eta(Nz^{\ell}-c)}p^i-v\mathcal{L}^{\ell}\big(\Delta(z,q)\big)\Big),
\end{align*}
where
\begin{align*}
\mathcal{L}^{\ell}\big(\Delta(z,q)\big)\!\!:=\!\!1\!+\!\eta\sum_{i=1}^{N}\gamma_{i}^{-1}\bigg(1\!-\!\exp\bigg(\!\!-\!\!\gamma_{i}\bigg(z\!+\!\Delta^{i,\ell}(z,q)\mathbf{1}_{\{\Delta^{i,\ell}(z,q)=\underline{\Delta}^i(z,q)\}}\!+\!\sum_{j=1}^{K}\omega_{j}\Delta^{i,\ell}(z,q)\mathbf{1}_{\{\Delta^{i,\ell}(z,q)\in K_{j}\}}\!\bigg)\!\bigg)\!\bigg).
\end{align*}
We now provide the optimal incentives corresponding to the solution of \eqref{HJB equation}.
\begin{lemma}\label{Lemma maximizers HJB}
Assume $\delta_{\infty}$ is large enough so that the condition of {\rm Lemma \ref{Lemma exchange Hamiltonian}} is verified. The optimisers in the supremum appearing in {\rm PDE} \eqref{HJB equation} are given, for any $(t,q)\in[0,T]\times\Qc^N$, by
\begin{align*}
&z^{\star,a}(t,q):=\frac{1}{N}\Bigg(c+\frac{1}{\eta}\mathrm{log}\bigg(\frac{v(t,q)}{\sum_{i\in \mathcal{G}}v(t,q\ominus_i 1)}\bigg)+\frac{1}{\eta}\mathrm{log}\bigg(\frac{k\varpi}{k\varpi+\sigma\eta}\mathrm{Card}(\mathcal{G})\bigg(1+\eta\sigma\sum_{i=1}^{N}\frac{1}{k\varpi+\sigma\gamma_{i}}\bigg)\bigg) \Bigg), \\
& z^{\star b}(t,q):=\frac{1}{N}\Bigg(c+\frac{1}{\eta}\mathrm{log}\bigg(\frac{v(t,q)}{\sum_{i\in \mathcal{G}}v(t,q\oplus_i 1)}\bigg)+\frac{1}{\eta}\mathrm{log}\bigg(\frac{k\varpi}{k\varpi+\sigma\eta}\mathrm{Card}(\mathcal{G})\bigg(1+\eta\sigma\sum_{i=1}^{N}\frac{1}{k\varpi+\sigma\gamma_{i}}\bigg)\bigg) \Bigg),\\
&z^{\star,S,i}(q):=-\sum_{j=1}^{N}\mu_{i,j}\gamma_{j}q^{j},\;  \forall i\in\{1,\dots,N\},
\end{align*}
where for all $(i,j)\in\{1,\dots,N\}^2$
\begin{align*}
\mu_{i,j}:=-\eta\kappa\prod_{k\in\{1,\dots,N\}\setminus\{i,j\}}\gamma_{k},\; \text{\rm if $i\neq j$},\; \mu_{i,i}:=\kappa \Bigg(\prod_{j\in\{1,\dots,N\}\setminus\{i\}}\gamma_{j}+\eta\sum_{j\in\{1,\dots,N\}\setminus\{i\}}\prod_{k\in\{1,\dots,N\}\setminus\{i,j\}}\gamma_{k}\Bigg),
\end{align*}
with
\begin{align*}
\kappa^{-1}:=\prod_{i=1}^{N}\gamma_{i}+\eta\sum_{j=1}^{N}\prod_{k\in\{1,\dots,N\}\setminus\{j\}}\gamma_{k},\; \text{\rm and}\; \mathcal{G}:=\Big\{i\in \{1,\dots,N\}: \gamma_{i}=\max_{j\in\{1,\dots,N\}}\gamma_j\Big\}.
\end{align*}
\end{lemma}

\medskip
The optimisers $z^{\star,a}$ and $,z^{\star,b}$ are only functions of time and the inventory of the agents. Moreover, they are very similar and share common properties with the optimal incentives $z^{\star,a}$, and $z^{\star,b}$ for the single market maker case in \cite{el2018optimal}: for example, for small inventories, they are decreasing function of the risk aversion parameters $\gamma^i$. However, the dependence on the number of market makers and their risk aversion is represented by the term
\begin{align*}
\frac{1}{\eta}\mathrm{log}\bigg(\frac{k\varpi}{k\varpi+\sigma\eta}\text{Card}(\mathcal{G})\bigg(1+\eta\sigma\sum_{i=1}^{N}\frac{1}{k\varpi+\sigma\gamma_{i}}\bigg)\bigg).    
\end{align*}
It is an increasing function of $N$ and $\varpi$, which implies that when we increase the number of market makers and $\varpi$, this term decreases the average spread. 

\medskip
Notice also that the optimal $z^{\star,S}$ depends on a weighted combination of all risk aversions and inventory processes of the agents. This is discussed more in Section \ref{Section comparaison}, where we also present our numerical results. 

\subsection{Change of variable and verification theorem}

Substituting the optima given by Lemma \ref{Lemma maximizers HJB}, PDE \eqref{HJB equation} boils down to
\begin{align}\label{PDE non-linear before change of variable}
\begin{cases}
    \displaystyle    \partial_{t}v(t,q)\!+\!v(t,q)C^{S}(q)\!-v(t,q)C\!\!\sum_{j\in\{a,b\}}\!\!\bigg(\frac{v(t,q)}{\sum_{i\in \mathcal{G}}v(t,q\ominus_i \phi(j))\mathbf{1}_{\{\phi(j)q^{i}>-\overline{q}\}}}\bigg)^{\frac{k\varpi}{\sigma\eta}}\!\!=0,\; (t,q)\in[0,T)\times\mathcal{Q}^{N},  \\
    \displaystyle    v(T,q)=-1,\; q\in\mathcal{Q}^N,
    \end{cases}
\end{align}
where we defined
\begin{align*}
C:=&\ A\exp\bigg(-\frac{k}{\sigma}\bigg(c\big(1-\varpi\big)-\frac{\varpi}{\eta}\log\bigg(\frac{k\varpi}{k\varpi+\eta\sigma}\mathrm{Card}\big(\mathcal{G}\big)\bigg(1+\eta\sigma\sum_{i=1}^{N}\frac{1}{k\varpi+\sigma\gamma_{i}}\bigg)\bigg)\\
& +\varpi\sum_{i=1}^{N}\gamma_{i}^{-1}\log\bigg(1+\frac{\sigma\gamma_{i}}{k\varpi}\bigg)\bigg)\bigg) \frac{\sigma\eta}{k\varpi+\sigma\eta}\bigg(1+\eta\sigma\sum_{i=1}^{N}\frac{1}{k\varpi+\sigma\gamma_{i}}\bigg),\\
C^S(q):=&\ \sum_{i=1}^{N}\frac{\eta}{2}\sigma^{2}\gamma_{i}\bigg(q^{i}-\sum_{j=1}^{N}\mu_{i,j}\gamma_{j}q^{j}\bigg)^{2}+\frac{\eta^{2}\sigma^{2}}{2}\bigg(\sum_{i=1}^{N}\sum_{j=1}^{N}\mu_{i,j}\gamma_{j}q^{j}\bigg)^{2}.
\end{align*}
\begin{lemma}\label{Lemma Cauchy Lipschitz}
There exists a unique bounded solution to \eqref{PDE non-linear before change of variable}, which is also negative.
\end{lemma}
The solution of \eqref{PDE non-linear before change of variable} will be linked to the value function \eqref{Definition reduced exchange control problem} using a verification argument in the next section. Note that, if the agents have different risk aversion parameters, one market maker is both best bid and best ask. Indeed, when a market maker is simultaneously the single best bid and best ask at some time $t\in[0,T)$, the HJB equation reduces to the following linear PDE
\begin{align*}
&0=\partial_{t}u-u(t,q)\tilde{C}^{S}(q)+\tilde{C}u(t,q\oplus_i 1)\mathbf{1}_{\{q^{i}<\overline{q}\}}+\tilde{C}u(t,q\ominus_i 1)\mathbf{1}_{\{q^{i}>-\overline{q}\}},
\end{align*}
where $u:=(-v)^{-\frac{k\varpi}{\sigma\eta}}$, $\tilde C^{S}(q):=(k\varpi)/(\sigma\eta)C^{S}(q)$, and  $\tilde C:=(k\varpi)/(\sigma\eta)C$. As the other inventories are fixed when the $i-$th market maker is quoting, we obtain a tridiagonal matrix similar to the one in {\rm\cite{el2018optimal}}, indexed by $q^{i}\in \mathcal{Q}$. We emphasize that such form is valid only at the fixed time $t$.

\medskip
We conclude with the following verification theorem, which leads to the description of a unique optimal contract to be proposed by the exchange to each market maker.
\begin{theorem}\label{Theorem Verification}
Assume that $\delta_{\infty}\geq \Delta_{\infty}$, as defined in {\rm Lemma \ref{Lemma exchange Hamiltonian}}, and let $v$ be the unique solution to \eqref{PDE non-linear before change of variable} given by {\rm Lemma \ref{Lemma Cauchy Lipschitz}}. Then, for any $i\in\{1,\dots,N\}$, the optimal contract for the $i-$th agent is given by
\begin{align}
\xi^{\star,i}:=\hat{y}_{0}^i+\int_{0}^{T}Z_{r}^{\star,a}\mathrm{d}N_{r}^{a}+Z_{r}^{\star,b}\mathrm{d}N_{r}^{b}+Z_{r}^{\star,S,i}\mathrm{d}S_{r}+\bigg(\frac{1}{2}\sigma^{2}\gamma_{i}(Z_{r}^{\star,S,i}+Q_{r}^{i})^{2}-H^{i}\big(\Delta(Z^{\star}_{r},Q_{r}),Z_{r}^{\star},Q_{r}\big)\bigg)\mathrm{d}r,
\end{align}
where for any $r\in[0,T]$, $Z_{r}^{\star,S}:=z^{\star,S}(r,Q_{r^{-}})$, $Z_{r}^{\star,a}:=z^{*a}(r,Q_{r^{-}})$, $Z_{r}^{\star,b}=z^{\star,b}(r,Q_{r^{-}})$, and we note $Z^\star_r:=\big(Z_{r}^{\star,a},Z_{r}^{\star,b},Z_{r}^{\star,S}\big)$. Moreover, the optimal equilibrium is given by $\big(\Delta(Z^\star_r,Q_r)\big)_{r\in[0,T]}$, see {\rm Corollary~\ref{Corollary best response agent}}.
\end{theorem}

\subsection{Discussion}\label{Section Discussion}

\subsubsection{Switching policy}

We want to determine which market maker is the best one at the beginning of the trading period. For any $(i,j)\in\{1,\dots,N\}^2$ such that $i\neq j$, the $i-$th market maker has the best ask quotation at time $t\in[0,T]$ if and only if
\begin{align*}
& -Z_{t}^{\star,a}+\frac{1}{\gamma_{i}}\mathrm{log}\bigg(1+\frac{\sigma\gamma_{i}}{k\varpi}\bigg)<-Z_{t}^{\star,a}+\frac{1}{\gamma_{j}}\mathrm{log}\bigg(1+\frac{\sigma\gamma_{j}}{k\varpi}\bigg).
\end{align*}
Since the term $\frac{1}{\gamma_{i}}\mathrm{log}\big(1+\frac{\sigma\gamma_{i}}{k\varpi}\big)$ is a decreasing function of $\gamma_{i}$, we conclude that the $i-$th market maker trades first if and only if $\gamma_{i}=\max_{j\in\{1,\dots,N\}}\gamma_{j}$, and we have 
\begin{align*}
& \bigg(\frac{v(t,Q_t)}{\sum_{j\in \mathcal{G}}v(t,Q_t\ominus_j 1)}\bigg)^{\frac{k\varpi}{\sigma\eta}} =\bigg(\frac{v(t,Q_t)}{v(t,Q_t\ominus_i 1)}\bigg)^{\frac{k\varpi}{\sigma\eta}}.
\end{align*} 

\medskip
We now define when there is a switching between two agents on the ask side. The $N-1$ other marker makers (recall that $j\neq i$) will place their quotes among the open covering of $[0,\delta_\infty]$. Assume that $\Delta^{i,a}(Z^{\star}_t,Q_t)\in K_u$, for some $u\in\{1,\dots,K\}$. Then
\begin{align}\label{Equivalent condition switching policy}
\Delta^{i,a}(Z^{\star}_t,Q_t)>\Delta^{j,a}(Z^{\star}_t,Q_t) \Longleftrightarrow -Z_{t}^{\star,a}+\frac{1}{\gamma_{i}}\mathrm{log}\bigg(1+\frac{\sigma\gamma_{i}}{k\varpi}\bigg) >\frac{1}{\omega_{u}}\bigg(-Z_{t}^{\star,a}+\frac{1}{\gamma_{j}}\mathrm{log}\bigg(1+\frac{\sigma\gamma_{j}}{k\varpi}\bigg)\bigg),
\end{align}
which can be rewritten as
\begin{align*}
\mathrm{log}\bigg(\frac{u(t,Q_{t^{-}})}{u(t,Q_{t^{-}}\ominus_i 1)}\bigg)>&\ \frac{k\varpi}{\sigma}\bigg(\mathrm{log}\bigg(\frac{k\varpi}{k\varpi+\sigma\eta}\bigg(1+\eta\sigma\sum_{i=1}^{N}\frac{1}{k\varpi+\sigma\gamma_{i}}\bigg)\bigg)+ c\bigg) + \frac{kN\varpi\omega_{u}}{(\omega_{u}-1)\sigma\gamma_{i}}\mathrm{log}\bigg(1+\frac{\sigma\gamma_{i}}{k\varpi}\bigg) \\
&+\frac{kN\varpi}{\sigma\gamma_{j}(\omega_{u}-1)}\mathrm{log}\bigg(1+\frac{\sigma\gamma_{j}}{k\varpi}\bigg).
\end{align*}
The right--hand side of the inequality is an increasing function of $\varpi$ and a decreasing function of $\omega_u$ and the volatility $\sigma$. These results are completely symmetric for the bid side. Following, \cite{el2018optimal} the previous equations shows that there is a switching between market makers on the ask side when the $i$--th market maker holds a sufficiently negative inventory. This is because he is willing to attract bid order to mean revert his inventory towards zero. Hence, he proposes a lower spread on the bid side, and a higher spread on the ask side to discourage ask orders. Symmetric conclusion holds for the bid side.

\subsubsection{On the number of market makers}

As the value function $v$ depends implicitly on the number of market makers through the terms $C^S(q)$ and $C$, we cannot directly maximize it with respect to $N$. However, using additional assumptions and working in an asymptotic setting, we show in this section that the optimal number of market makers is finite. Numerical computations of $N$ will then be given in Section \ref{Section comparaison}. We use $N$ as a subscript to highlight the dependence of the functions with respect to the number of market makers.  

\medskip
Let us define 
\[
v_N (0,Q_0) = c(N_T^a + N_T^b) - Y_{T}^{y_{0},Z^\star,\Delta(Z^\star,Q)}\cdot 1_N.
\]
As the market makers' inventory mean revert toward zero, it is reasonable to study the behaviour of $v_N(t,Q_t)$ when $Q_t=0$. In that case, $\log\Big(\frac{v(t,q)}{\sum_{i=1}^N v(t,q\ominus_i 1)}\Big)=\log\Big(\frac{v(t,q)}{\sum_{i=1}^N v(t,q\oplus_i 1)}\Big)\approx 0$ and 
\begin{align*}
    Z_t^{\star,j} \underset{Q_t\rightarrow 0}{=} \frac{1}{N}\bigg( c + \frac{1}{\eta}\mathrm{log}\bigg(\frac{k\varpi}{k\varpi+\sigma\eta}N\bigg(1+\eta\sigma\sum_{i=1}^{N}\frac{1}{k\varpi+\sigma\gamma_{i}}\bigg)\bigg)\bigg), \; j\in \{a,b\}.
\end{align*}
By taking expectations, we obtain 
\begin{align*}
    \mathbb{E}^{\Delta(Z^\star,0)}\big[v_N(0,Q_0)\big] = \mathbb{E}^{\Delta(Z^\star,0)}\bigg[\int_0^T & \Big( -B_N\Big(\lambda_N^a\big(\Delta^a(Z_t^\star,0),0\big)+\lambda_N^b\big(\Delta^b(Z_t^\star,0),0\big)\Big)  +\mathcal{W}_N(Z_t^\star)\Big)\mathrm{d}t \bigg]. 
\end{align*}
where 
\begin{align*}
    & B_N=\frac{1}{\eta}\mathrm{log}\bigg(\frac{k\varpi}{k\varpi+\sigma\eta}N\bigg(1+\eta\sigma\sum_{i=1}^{N}\frac{1}{k\varpi+\sigma\gamma_{i}}\bigg)\bigg),\\
    &  \mathcal{W}_N(Z_t^\star) = \sum_{i=1}^N \big( H_N^i (\Delta(Z_t^\star,0) - \frac{\hat{y}_0^i}{T}\big)\Big).
\end{align*}
For the sake of simplicity, assume that all the risk aversion parameters $\gamma_i$ are of the same magnitude (i.e $\gamma_i = \gamma$ for all $i$) so that
\begin{align*}
    \sum_{i=1}^N H_N^i \big(\Delta(Z_t^\star,0),Z_t^\star,0\big) =& \sum_{j=a,b} N\frac{\sigma}{k\varpi + \sigma \gamma} \lambda_N^j\big(\Delta^j(Z_t^\star,0),0\big), \\
    B_N =&\ \frac{1}{\eta}\log \bigg(\frac{k\varpi}{k\varpi+\sigma\eta}N\bigg(1+\frac{N\eta\sigma}{k\varpi + \sigma \gamma}\bigg)\bigg), \\
    \lambda_N^j\big(\Delta^j(Z_t^\star,0),0\big)  =&\  A\exp\bigg(-\frac{k}{\sigma}\bigg(c(1-\varpi) -\varpi \bigg( - \frac{N}{\gamma}\log\bigg(1+\frac{\sigma\gamma}{k\varpi}\bigg) \\
   &\hspace{3.7em}+\frac{1}{\eta}\log \bigg(\frac{k\varpi}{k\varpi+\sigma\eta}N\bigg(1+\frac{N\eta\sigma}{k\varpi + \sigma \gamma}\bigg)\bigg) \bigg) \bigg)\bigg).
\end{align*}
We see that $\lambda_N^j\big(\Delta^j(Z_t^\star,0),0\big)\rightarrow_{N\rightarrow +\infty} 0$. This implies that a too high number of market makers with comparable risk aversion will, on average, decrease the liquidity available on the market and therefore decrease the profits of the platform.

\medskip
Finally, define the (same) reservation utility of the market makers as $\hat{y}_0^i=\frac{k\varpi}{\sigma}\log\big(w_N(0,Q_0)\big)$, where $w_N(0,Q_0)$ is the value function of the market maker when $\xi=0$ and every agent has the same risk aversion parameter. We obtain 
\begin{align*}
      \mathcal{W}_N(Z_t^\star) = N \bigg( \frac{\sigma}{k\varpi + \sigma\gamma} \sum_{j=a,b}\lambda_N^j\big(\Delta^j(Z_t^\star,0),0\big) - \frac{k\varpi}{\sigma\eta}\log\big(w_N(0,Q_0)\big)\bigg). 
\end{align*}
As the optimal market making solutions have a stationary behavior when $T$ is sufficiently large, we approximate the value function $w_N$ using a Taylor expansion with respect to $T$
\begin{align*}
    w_N (t,Q_t) \approx 1+ 2\hat{C}_N(T-t),
\end{align*}
where $\hat{C}_N \approx 2 \frac{\sigma}{k\varpi + \sigma\gamma} A\exp\big(-\frac{k}{\sigma}(c+\frac{\varpi N}{\gamma}\log(1+\frac{\sigma\gamma}{k\varpi})\big)$. Therefore, we obtain
\begin{align*}
    \mathcal{W}_N(Z^\star) \approx &\ 2 \frac{A\sigma}{k\varpi + \sigma\gamma} N \exp\Big(-\frac{k}{\sigma}\big(c+\frac{\varpi N}{\gamma}\log(1+\frac{\sigma\gamma}{k\varpi})\big)\Big) \\
    & \times \Bigg( \exp\Bigg(\frac{k\varpi}{\sigma}\Bigg(c+\frac{1}{\eta}\log\bigg(\frac{k\varpi}{k\varpi+\sigma\eta}N\bigg(1+\frac{N\eta\sigma}{k\varpi+\sigma\gamma}\bigg)\bigg)\Bigg)\Bigg) - \frac{k\varpi}{\sigma\eta}\Bigg).
\end{align*}
We finally set 
\begin{align*}
    \mathcal{S}(N) = &\mathcal{W}_N(Z^\star) -2\lambda_N^a\big(\Delta^a(Z^\star,0),0\big)B_N. 
\end{align*}
This function is differentiable with respect to $N$, and we obtain $\mathcal{S}(0)>0$. By decreasing of the exponential, $\lim_{N\rightarrow +\infty}\mathcal{S}(N)=0$. Moreover, thanks to simple but tedious computations, we have $\lim_{N\rightarrow 0^+} \frac{\partial \mathcal{S}(N)}{\partial N} >0$. These conditions guarantee the existence of, at least, one global maximum of the function $\mathcal{S}(N)$. \\

We observe that, for the same set of parameters than in Section \ref{Section comparaison}, the function $\mathcal{S}(N)$ attains it supremum for $N\approx 3$, which corresponds to the optimal number of market makers found numerically with the resolution of the HJB equation.

\subsubsection{On the choice of the weights}

Numerical experiments show, when $N\longrightarrow +\infty$, a decrease of the intensity of the market orders, and a slight increase of the average bid--ask spread, other parameters being fixed. However, an increase of $\varpi$ decrease the average best bid--ask spread, as well as the PnL of the platform, increase the total order flow and decrease the trading cost. Recall that there is a trade--off between an increase of the order flow, and the amount of incentive given to the market participants. Moreover, increasing the competition between market makers leads to an increase of their reservation utility, which is costly for the principal. Recall that we designed the aggregated intensity to be a decreasing function of a weighted sum of the spreads quoted by the agents. In practice, the intensity of arrival orders mainly depends on the best quote $\underline{\delta}$, that is to say for $j\in\{a,b\}$, and $t\in[0,T)$
\begin{align*}
&\lambda^{j}(\delta_{t}^{:,j},Q_t)=A\exp\bigg(-\frac{k}{\sigma}\Big(c+\varpi \sum_{i=1}^{N}\delta_{t}^{i,j}\mathbf{1}_{\{\delta_{t}^{i,j}=\underline \delta_{t}^{j} \}}+\sum_{i=1}^{N}\sum_{\ell=1}^{K}H_{\ell}\delta_{t}^{i,j}\mathbf{1}_{\{\delta_{t}^{i,j}\in K_{\ell} \}}\Big)\bigg)\approx A\exp\Big(-\frac{k}{\sigma}\big(c+\underline \delta_{t}^{j}\big)\Big).
\end{align*}
Assume that $\mathcal{G}=\{i\},$ $\varpi=\frac{1}{N}$. The optimal quotes in Theorem \ref{Theorem Verification} become, for $j\in\{a,b\}$
\begin{align*}
\underline \Delta^{j}(Z_{t}^{\star}\!,\!Q_t)\!=\!\frac{\sigma}{k}\mathrm{log}\bigg(\!\frac{u(t,Q_{t^{-}})}{u(t,Q_{t^{-}}\ominus_i \phi(j))}\!\bigg)\!\!+\!\!\frac{1}{\gamma_{i}}\mathrm{log}\bigg(\!1\!+\!\frac{\sigma\gamma_{i}N}{k}\bigg) \!\!-\!\!\frac{1}{N}\bigg(\!c\!+\!\frac{1}{\eta}\mathrm{log}\bigg(\frac{k}{k+\sigma\eta N}\!\bigg(\! 1\!+\!\eta\sigma\!\sum_{i=1}^{N}\frac{N}{k+\sigma\gamma_{i}N}\!\bigg)\!\bigg)\!\bigg).
\end{align*}
Hence, when the number of market maker increases, the last term corresponding to the incentive given by the principal vanish to zero and we are left with, for $j\in\{a,b\}$
\begin{align*}
& \underline\Delta^{j}(Z_{t}^{\star},Q_t)\approx \frac{\sigma}{k}\mathrm{log}\bigg(\frac{u(t,Q_{t^{-}})}{u(t,Q_{t^{-}}\ominus_i \phi(j))}\bigg)+\frac{1}{\gamma_{i}}\mathrm{log}\bigg(1+\frac{\sigma\gamma_{i}N}{k}\bigg).
\end{align*}
It therefore converges toward the form of spread given when there is no contract, but with a different value function.

\subsubsection{On the form of the incentives}

The quantities $z^{\star, j}$, $j\in\{a,b\}$, defined in Lemma \ref{Lemma maximizers HJB}, are decreasing function of the number of market makers. Hence, the principal is limited in the amount of incentives he can provide to the agents. This can be viewed as a cake whose size increase slower than the number of people who eats it. Hence, each market maker receive less incentive to decrease their spread in our case of a uniform incentive and an increasing number of market makers.

\medskip
About the risk aversion of the additional market makers, adding a player with a small risk aversion increase the quantity $z^{\star, j}$, $j\in\{a,b\}$. This means that adding a less risk adverse player increase the capacity of the principal to offer incentive to reduce the average spread and conversely. 

\medskip
We have found processes $Z^{\star,a},Z^{\star, b},Z^{\star,S}$ fixed by the principal in order to build optimal contracts for every market makers. The assumption that the exchange chooses a priori the same incentives on the arrival orders for each market maker is quite natural, since in practice the principal may not know the risk aversions of each market maker. When a market maker is simultaneously best bid and best ask, we recover the result from \cite[Proposition 4.1]{el2018optimal}, that the terms $-\mathrm{log}\big(\frac{u(t,Q_{t^{-}})}{u(t,Q_{t^{-}}\ominus_i 1)}\big)$ and $-\mathrm{log}\big(\frac{u(t,Q_{t^{-}})}{u(t,Q_{t^{-}}\oplus_i 1)}\big)$ are roughly proportional to, respectively, $Q_{t^{-}}^{i}$ and $-Q_{t^{-}}^{i}$. The interpretation is the same: the exchange provides incentives to the market makers to keep their inventory not too large. 

\medskip
An interesting difference comes from the integrals $\int_{0}^{T}Z_{r}^{\star,S,i}\mathrm{d}S_{r}$ . As in \cite{el2018optimal}, it is still understood as a risk sharing term. However, each of the $Z^{\star,S,i}$ is a weighted function of both $\gamma_i$ and the other risk aversions $\gamma_j$, $j\neq i$. Indeed, when the risk aversion of the $i-$th market maker increases, $Z^{\star,S,i}$ decreases. When the risk aversions of the $N-1$ other market makers increase, $Z^{\star,S,i}$ increases and conversely.

\subsubsection{On the taker cost policy}
When the $i-$th market maker is simultaneously best bid and best ask, the exchange can fix a relevant value of the taker cost $c$ as in \cite{el2018optimal}. From numerical computations
\begin{align*}
\frac{u(t,q)^{2}}{u(t,q\oplus_i 1)u(t,q\ominus_i 1)}\approx 1,\; \text{for all } (t,q)\in [0,T]\times\mathbb{Z}^N.
\nonumber
\end{align*}
Hence the exchange may fix in practice the transaction cost $c$ so that the average best spread is close to one tick by setting 
\begin{align*}
c \approx -\frac{1}{2N}\text{Tick}-\frac{1}{\eta N}\mathrm{log}\bigg(\frac{k\varpi}{k\varpi+\sigma\eta}\bigg(1+\eta\sigma\sum_{i=1}^{N}\frac{1}{k\varpi+\sigma\gamma_{i}}\bigg)\bigg)+\frac{1}{\gamma_{i}N}\mathrm{log}\bigg(1+\frac{\sigma \gamma_{i}}{k\varpi}\bigg).
\end{align*}
When $\sigma\eta/k\varpi$, and $\sigma\gamma_{i}/k\varpi$ are small enough for all $ i\in\{1,\dots,N\}$, this equation reduces to 
\begin{align*}
c\approx \frac{1}{N}\bigg(\frac{\sigma}{k\varpi}-\frac{1}{2}\text{Tick}\bigg).
\end{align*}
We therefore find a similar formula to the one in the case $N=1$, and notice that it is a decreasing function of the number of market makers, with $\varpi=\frac{1}{N}$. As $\sigma$ and $k$ can be estimated in practice using market data, this is a particularly useful rule of thumb to determine the taker cost $c$. However, when one market maker is the best bid and another one is the best ask, the approximation $\frac{u(t,q)^{2}}{u(t,q\oplus_i 1)u(t,q\ominus_i 1)}\approx 1$ is no longer valid. Hence, the exchange has the choice either to stay with the previous rule of thumb, or to monitor a time-dependent taker cost given by
\begin{align*}
& c(t,q) \!\approx \!-\!\frac{1}{2N}\text{Tick}\!-\!\frac{1}{\eta N}\!\bigg(\!\mathrm{log}\!\Big(\!\frac{u(t,q)^{2}}{u(t,\!q\!\ominus_i\! 1)u(t,\!q\!\oplus_j\!1)}\!\Big)\!+\!\mathrm{log}\Big(\!\frac{k\varpi}{k\varpi\!+\!\sigma\eta}\!\big(1\!+\!\eta\sigma\!\sum_{i=1}^{N}\!\frac{1}{k\varpi\!+\!\sigma\gamma_{i}}\!\big)\!\Big)\!\bigg)\!\!+\!\! \frac{1}{\gamma_{i}N}\mathrm{log}\Big(\!1\!\!+\!\!\frac{\sigma \gamma_{i}}{k\varpi}\!\Big).
\end{align*}
where the $i-$th agent is the best ask, and the $j-$th is the best bid.

\section{Impact of the presence of several market makers} \label{Section comparaison}

In this section, we compare our results with the ones given in \cite{el2018optimal}. 

\subsection{One market maker}
As a sanity check, we want to recover the results of \cite{el2018optimal}. We take the same numerical values for the parameters, namely $T = 600s$ for an asset with volatility $\sigma = 0.3$ Tick.$s^{-1/2}$ (unless specified differently). Market orders arrive according to the intensities described in Section \ref{Section The Model}, with $A = 1.5s^{-1}$ and $k = 0.3s^{-1/2}$. We have $\overline{q}=50$, $\gamma=0.01, \eta=1$, $c=0.5\text{Tick}$ and $\varpi=1$. We directly present the results of our model

\begin{figure}[!ht] 
\begin{minipage}[c]{.46\linewidth}
     \begin{center}
             \includegraphics[width=7cm]{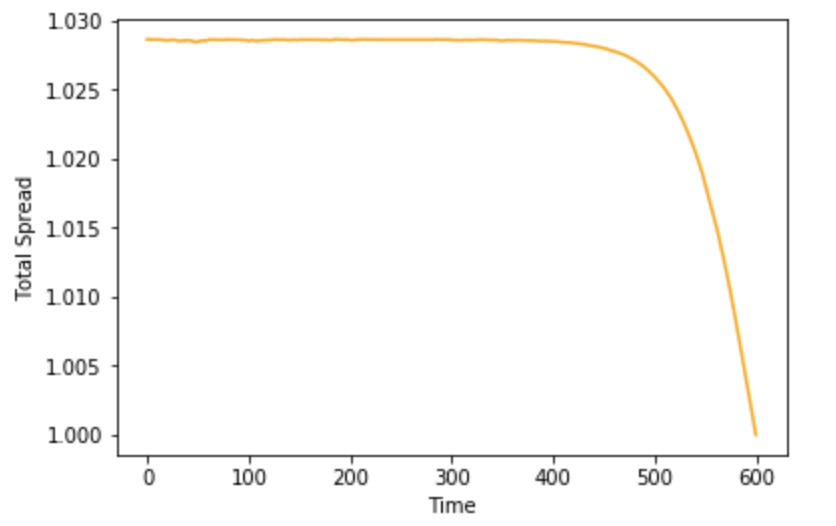}
             \caption{Total spread for 1 market maker}\label{total spread 1MM}
         \end{center}
   \end{minipage} \hfill
   \begin{minipage}[c]{.46\linewidth}
    \begin{center}
            \includegraphics[width=7cm]{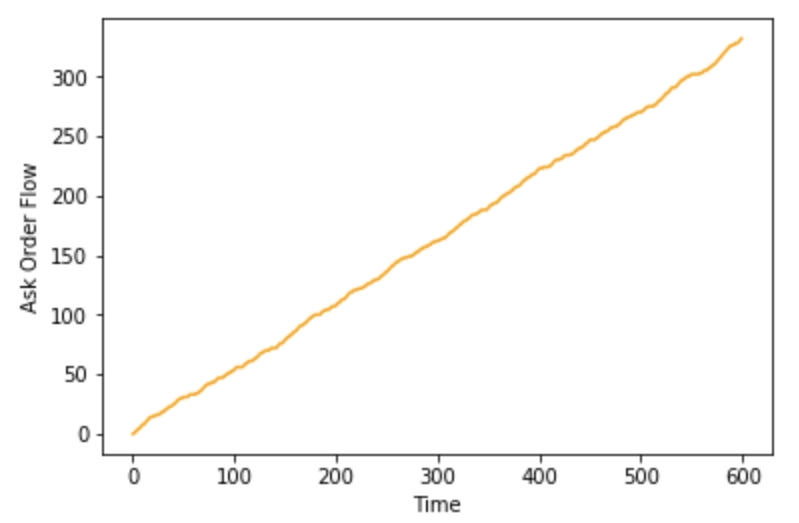}
            \caption{Ask order flow for 1 market maker}\label{ask order 1MM}
        \end{center}
\end{minipage}
\end{figure}

\begin{figure}[!ht]
\begin{minipage}[c]{.46\linewidth}
     \begin{center}
             \includegraphics[width=7cm]{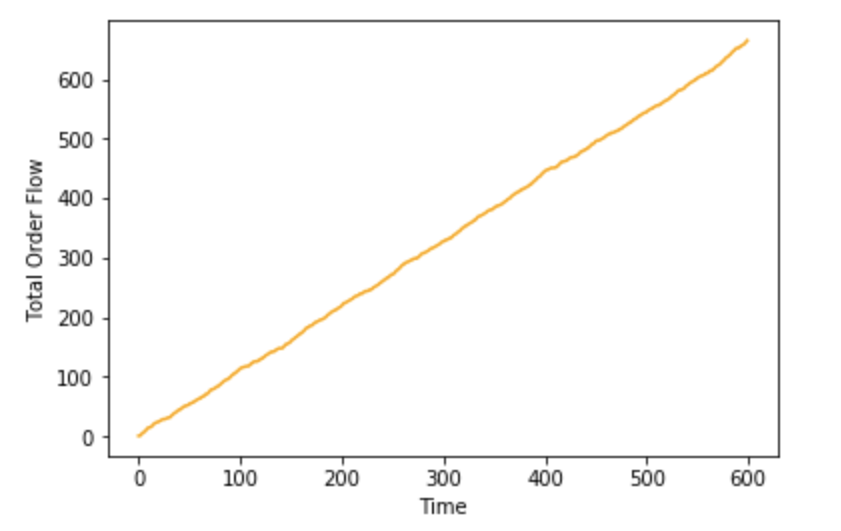}
             \caption{Total order flow for 1 market maker}\label{total order flow 1MM}
         \end{center}
   \end{minipage} \hfill
   \begin{minipage}[c]{.46\linewidth}
    \begin{center}
            \includegraphics[width=7cm]{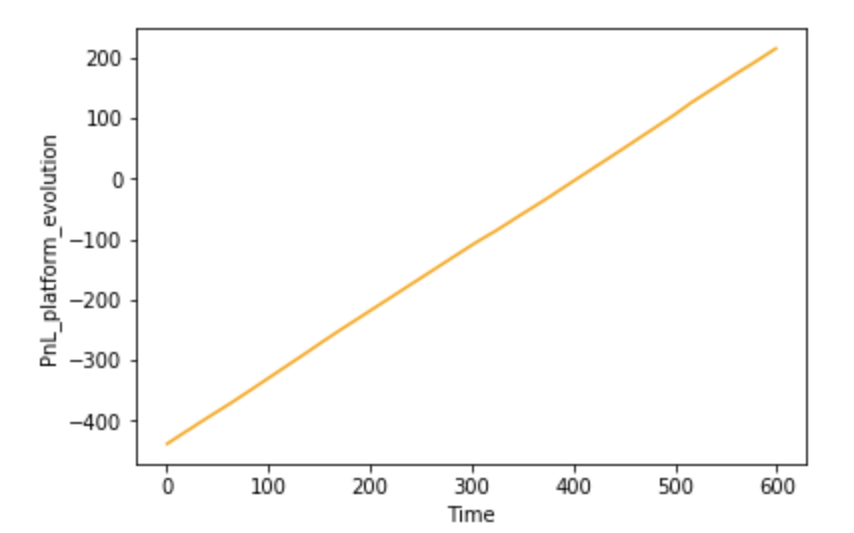}
            \caption{PnL of the exchange for 1 market maker}\label{PnL Exchange 1MM}
        \end{center}
\end{minipage}
\end{figure}

\begin{figure}[!ht]
\begin{center}
    \includegraphics[width=7cm]{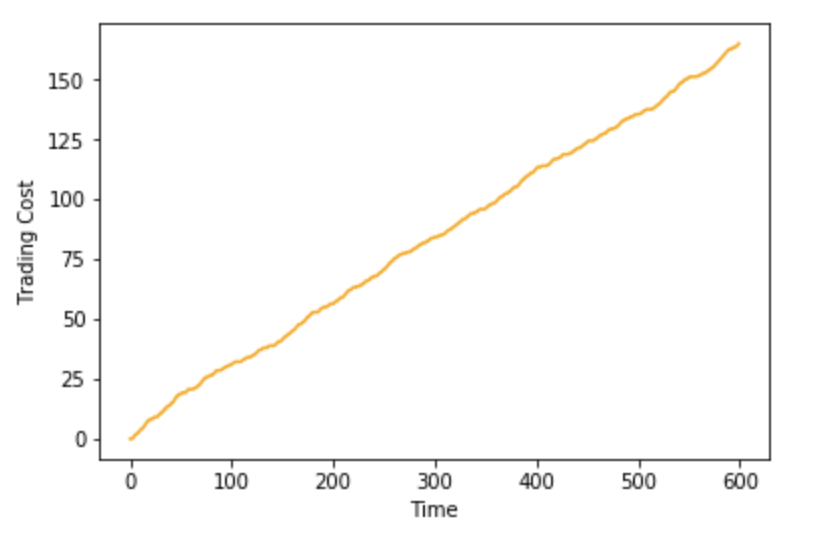}
    \caption{Trading cost for 1 market maker}\label{trading cost 1MM}
\end{center}
\end{figure}

We see in Figures \ref{total spread 1MM} to \ref{trading cost 1MM} that we recover the results obtained in \cite{el2018optimal}. We now turn to the case $N\geq 2$. 

\subsection{Two market makers}

We first begin with the average spread in the case $N=2$, with $\varpi=\frac{1}{2}$. The brackets in the title of the figures denote the set of risk aversion of the agents.

\begin{figure}[!ht]
\begin{minipage}[c]{.46\linewidth}
     \begin{center}
             \includegraphics[width=7cm]{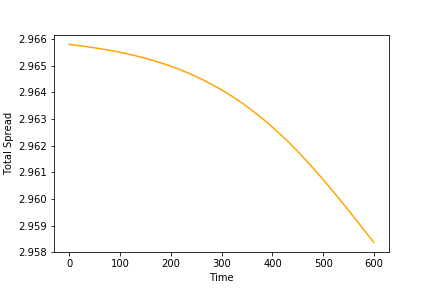}
             \caption{Total spread for $N=2$, $[0.01,0.001]$}\label{TotalSpread2MMwmin050001001}
         \end{center}
   \end{minipage} \hfill
   \begin{minipage}[c]{.46\linewidth}
    \begin{center}
            \includegraphics[width=7cm]{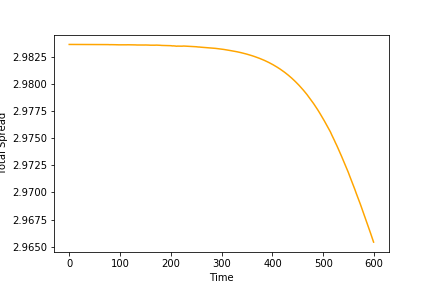}
            \caption{Total spread for $N=2$, $[0.01,0.01]$}\label{TotalSpread2MMwmin05001}
        \end{center}
\end{minipage}
\end{figure}

\begin{figure}[!ht]
\begin{center}
    \includegraphics[width=7cm]{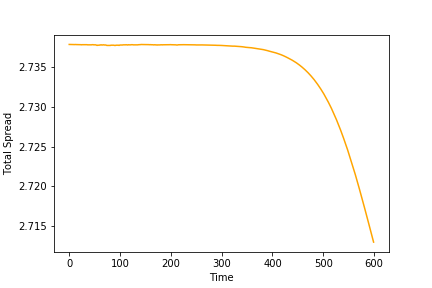}
    \caption{Total spread for $N=2$, $[0.01,0.1]$}\label{TotalSpread2MMwmin0500101}
\end{center}
\end{figure}
We can see in Figures \ref{TotalSpread2MMwmin050001001}, \ref{TotalSpread2MMwmin05001} and \ref{TotalSpread2MMwmin0500101} an increase of the total spread compared to the case $N=1$. As explained in Section \ref{Section Discussion}, this is due to the fact that the quantities $z^{\star j},j=a,b$ are decreasing function of $N$. Hence the incentive given to each market maker is less important than in the case $N=1$. In addition to this, adding a market maker with a higher risk aversion decrease the total spread and conversely. 

\medskip
Such spread induces a decrease of total order flow, see Figure \ref{TotalOrderFlow2MMwmin05001}, compared to the case $N=1$. For sake of simplicity we only present the results for two market makers with same risk aversion.

\begin{figure}[!ht]
\begin{minipage}[c]{.46\linewidth}
     \begin{center}
             \includegraphics[width=7cm]{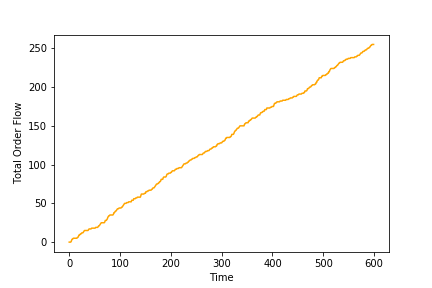}
             \caption{Total order flow for $N=2$, $[0.01,0.01]$}\label{TotalOrderFlow2MMwmin05001}
         \end{center}
   \end{minipage} \hfill
   \begin{minipage}[c]{.46\linewidth}
    \begin{center}
            \includegraphics[width=7cm]{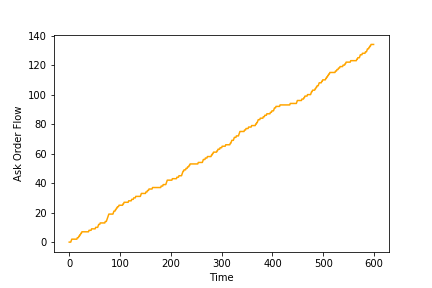}
            \caption{Ask order flow for $N=2$, $[0.01,0.01]$}\label{AskOrderFlow2MMwmin05001}
        \end{center}
\end{minipage}
\end{figure}
Similar results occurs for different risk aversion parameters, except that the decrease of order flow is less important with a second market maker with a higher risk aversion parameter and conversely. This also has an impact on the trading cost and the PnL of the platform, as it can be seen in Figures \ref{PnLPlatform2MMwmin05001},\ref{TradingCost2MMwmin05001}.

\begin{figure}[!ht]
\begin{minipage}[c]{.46\linewidth}
     \begin{center}
             \includegraphics[width=7cm]{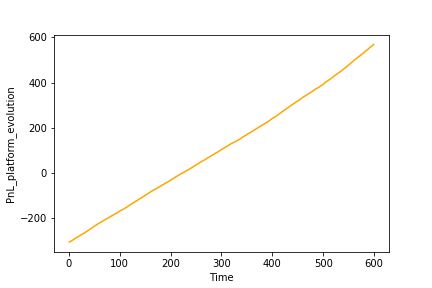}
             \caption{PnL of the exchange for $N=2$, $[0.01,0.01]$}\label{PnLPlatform2MMwmin05001}
         \end{center}
   \end{minipage} \hfill
   \begin{minipage}[c]{.46\linewidth}
    \begin{center}
            \includegraphics[width=7cm]{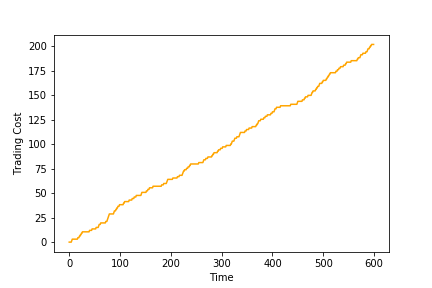}
            \caption{Trading cost for $N=2$, $[0.01,0.01]$}\label{TradingCost2MMwmin05001}
        \end{center}
\end{minipage}
\end{figure}

We can see in Figure \ref{TradingCost2MMwmin05001} an increase of the trading cost due to a mixed effect of the decrease of order flow, and an increase of the total spread, see Figures \ref{TotalSpread2MMwmin050001001},\ref{TotalSpread2MMwmin05001} and \ref{TotalOrderFlow2MMwmin05001}. However, we see an increase in the PnL of the exchange, mainly due to the fact that the reservation utility for every agent $\hat{y}_{0}^{i}:=\frac{k\varpi}{\sigma}\mathrm{log}\big(u(0,Q_{0})\big)\; i\in\{1,\dots,N\}$ is less important than in the case $N=1$. 

\subsection{Five market makers}
This case aims at illustrating what happens when we increase again the number of market makers. For sake of simplicity we only illustrate the case of market makers having the same risk aversion parameter equal to 0.01.

\begin{figure}[!ht]
\begin{minipage}[c]{.46\linewidth}
     \begin{center}
             \includegraphics[width=7cm]{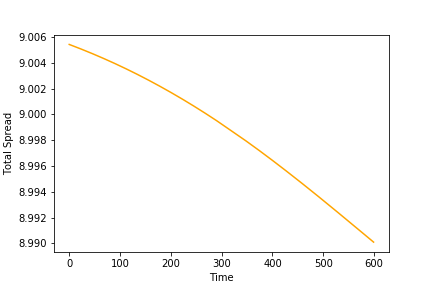}
             \caption{Total spread for $N=5$}\label{TotalSpread5MMunweighted}
         \end{center}
   \end{minipage} \hfill
   \begin{minipage}[c]{.46\linewidth}
    \begin{center}
            \includegraphics[width=7cm]{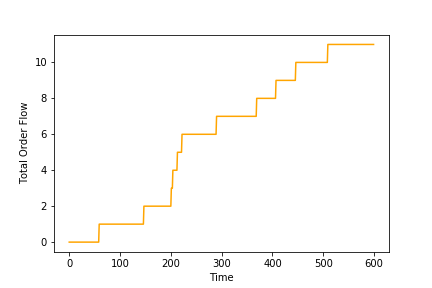}
            \caption{Total order flow for $N=5$ }\label{TotalOrderFlow5MMunweighted}
        \end{center}
\end{minipage}
\end{figure}

\begin{figure}[!ht]
\begin{minipage}[c]{.46\linewidth}
     \begin{center}
             \includegraphics[width=7cm]{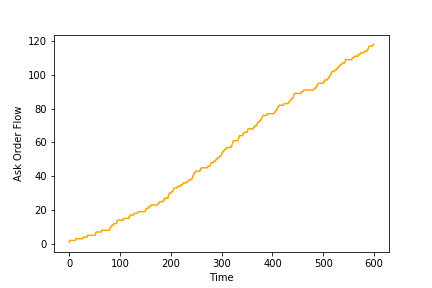}
             \caption{Ask order flow for $N=5$ }\label{WeightedAskOrderFlow2MMwmin05001}
         \end{center}
   \end{minipage} \hfill
   \begin{minipage}[c]{.46\linewidth}
    \begin{center}
            \includegraphics[width=7cm]{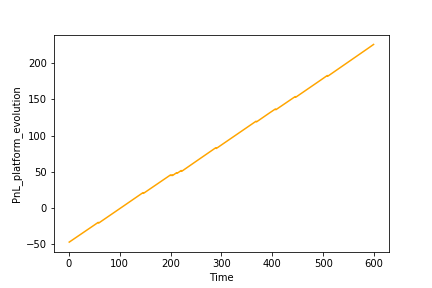}
            \caption{PnL of the exchange for $N=5$ }\label{PnLPlatform5MMunweighted}
        \end{center}
\end{minipage}
\end{figure}

\begin{figure}[!ht]
\begin{minipage}[c]{.46\linewidth}
\begin{center}
    \includegraphics[width=7cm]{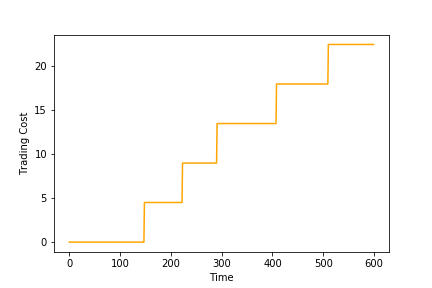}
    \caption{Trading cost for $N=5$}\label{TradingCost5MMunweighted}
\end{center}
   \end{minipage} \hfill
   \begin{minipage}[c]{.46\linewidth}
   \begin{center}
    \includegraphics[width=7cm]{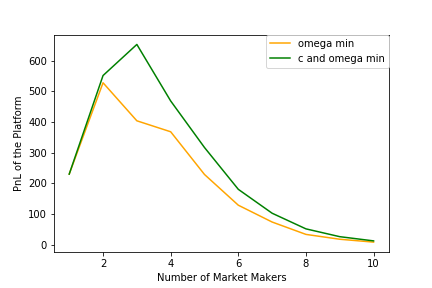}
    \caption{Evolution of the PnL of the platform with the number of market makers}\label{PlatformComparisonPnL}
\end{center}
   \end{minipage}
   
\end{figure}

As expected, we obtain in Figure \ref{TotalSpread5MMunweighted} a higher total spread, which implies a decrease of the order flow, see Figure \ref{TotalOrderFlow5MMunweighted}. However, in Figure \ref{PnLPlatform5MMunweighted}, the PnL of the platform has decreased compared to the case $N=2$. This means that it is not optimal for the platform to attract an infinite number of market makers. We conclude these numerical experiments with Figure \ref{PlatformComparisonPnL} showing how the PnL of the platform evolve with the number of market makers. We emphasise here that what is important is not the risk aversion of the market makers added to the market: this has an impact on the PnL of the platform but not on the trend of the graph. Hence, we add market makers with the same risk aversion equals to $0.01$. 

\medskip
In Figure \ref{PlatformComparisonPnL}, there are two different plots. The orange one is with $\varpi=\frac{1}{N}$ and $c$ the taker cost being fixed. The red one is with $\varpi=\frac{1}{N}$ and $c=\frac{1}{N}\big(\frac{\sigma}{k}-\frac{1}{2}\text{Tick}\big)$ as stated in the previous section. We can see that without an optimal taker cost policy, the optimal number of market makers for the platform is at $N=2$, other parameters being fixed. However, with an optimal policy, the platform is encouraged to add another market maker to increase its PnL. It is also worth noting that in both cases, the platform can add up to 4 market makers and still have a PnL higher than in the case $N=1$. 

\begin{appendix}
\section{Appendix}

\subsection{Dynamic programming principle}

For any $i\in\{1,\dots,N\}$, any $\mathbb{F}-$predictable stopping times $\tau$ taking values in $[0,T]$, any admissible contract vector $\xi\in\mathcal{C}$, any $2(N-1)-$dimensional $\F-$predictable process $\delta^{-i}$, bounded by $\delta_\infty$, and for all $\delta \in \mathcal{A}^i(\delta^{-i})$, we define
\begin{align*}
&J^{i}(\xi^i,\tau,\delta,\delta^{-i}):=\mathbb{E}_{\tau}^{\delta\otimes_{i}\delta^{-i}}\bigg[-\exp\bigg(-\gamma_{i}\int_{\tau}^{T}\delta_{u}^{a}\mathbf{1}_{\{\delta_{u}^{a}=\underline{\delta_u^a\otimes_i\delta_u^{a,-i}}\}}\mathrm{d}N_{u}^{a}+\delta_{u}^{b}\mathbf{1}_{\{\delta_{u}^{b}=\underline{\delta_u^b\otimes_i\delta_u^{b,-i}}\}}\mathrm{d}N_{u}^{b}+Q_{u}^{i}\mathrm{d}S_{u} \\
&\hspace{12.2em}+\sum_{\ell=1}^{K}\omega_{\ell}\Big(\delta_{u}^{a}\mathbf{1}_{\{\delta_{u}^{a}\in K_{\ell}\}}\mathrm{d}N_{u}^{a}+\delta_{u}^{b}\mathbf{1}_{\{\delta_{u}^{b}\in K_{\ell}\}}\mathrm{d}N_{u}^{b}\Big)\bigg) \exp\big(-\gamma_{i}\xi^{i}\big)\Bigg].
\end{align*}
We also define the family $\mathcal{J}_{\tau}^{i}:=(J^{i}(\xi^i,\tau,\delta,\delta^{-i}))_{\delta \in \mathcal{A}^{i}(\delta^{-i})}$. The continuation utility of the $i-$th market maker is defined by 
\begin{align}\label{Definition continuation utility value function agent}
V_{\tau}^{i}(\xi^i,\delta^{-i})=\underset{\delta\in \mathcal{A}^{i}(\delta^{-i})}{\text{ess sup}}J^{i}(\xi^i,\tau,\delta,\delta^{-i}).
\end{align}

\begin{lemma}\label{Lemma 6.1}
Let $\tau$ be a $\mathbb{F}$-predictable stopping time with values in $[t,T]$. Then, there exists a non--decreasing sequence $(\delta^{n})_{n\in \mathbb{N}}$ in $\mathcal{A}^{i}(\delta^{-i})$ such that $V_{\tau}^{i}(\xi^i,\delta^{-i})=\lim_{n\rightarrow +\infty}J^{i}(\xi^i,\tau,\delta^{n},\delta^{-i})$.
\end{lemma}
\begin{proof}
For $(\delta, \delta^\prime) \in \mathcal{A}^{i}(\delta^{-i})\times \mathcal{A}^{i}(\delta^{-i})$, we define
\begin{align*}
\overline{\delta}:=\delta\mathbf{1}_{\{J_{T}^{i}(\xi^i,\tau,\delta,\delta^{-i})\geq J_{T}^{i}(\xi^i,\tau,\delta',\delta^{-i}) \} }+\delta'\mathbf{1}_{\{J_{T}^{i}(\xi^i,\tau,\delta,\delta^{-i})\leq J_{T}^{i}(\xi^i,\tau,\delta',\delta^{-i}) \} }.
\end{align*}
We have $\overline{\delta}\in \mathcal{A}^{i}(\delta^{-i})$ and by definition of $\overline{\delta}$, $J^{i}(\xi^i,\tau,\overline{\delta},\delta^{-i})\geq \max\big(J^{i}(\xi^i,\tau,\delta,\delta^{-i}),J^{i}(\xi^i,\tau,\delta',\delta^{-i})\big)$. Hence, $\mathcal{J}_{\tau}^{i}$ is upward directed, and the required result follows from  \cite[Proposition VI.I.I p121]{neveu1975discrete}.
\end{proof}

\begin{lemma}\label{Lemma 6.2}
Let $t\in [0,T]$ and $\tau$ be an $\mathbb{F}-$predictable stopping time with values in $[t,T]$. Then
\begin{align*}
V_{t}^{i}(\xi^i,\delta^{-i})=\underset{\delta^{i}\in \mathcal{A}^{i}(\delta^{-i})}{\text{\rm ess sup}}\mathbb{E}_{t}^{\delta^{i}\otimes_{i}\delta^{-i}}\bigg[ & -\exp\bigg(\!-\!\gamma_{i}\!\int_{t}^{\tau}\delta_{u}^{a,i}\mathbf{1}_{\{\delta_{u}^{i,a}=\underline{\delta_u^a\otimes_i\delta_u^{a,-i}}\}}\mathrm{d}N_{u}^{a}\!+\!\delta_{u}^{i,b}\mathbf{1}_{\{\delta_{u}^{i,b}=\underline{\delta_u^b\otimes_i\delta_u^{b,-i}}\}}\mathrm{d}N_{u}^{b}\!+\! Q_{u}^{i}\mathrm{d}S_{u} \\
&+\sum_{\ell=1}^{K}\omega_{\ell}\Big(\delta_{u}^{i,a}\mathbf{1}_{\{\delta_{u}^{i,a}\in K_{\ell}\}}\mathrm{d}N_{u}^{a}+\delta_{u}^{i,b}\mathbf{1}_{\{\delta_{u}^{i,b}\in K_{\ell}\}}\mathrm{d}N_{u}^{b}\Big)\bigg)V_{\tau}^{i}(\xi^i,\delta^{-i})\bigg].  
\end{align*}
\end{lemma}

\begin{proof} Let $t\in [0,T]$ and fix an $\mathbb{F}-$predictable stopping time $\tau$ with values in $[t,T]$. To simplify the notations, we define for all $t\in [0,T]$ and $\delta\in \mathcal{A}$
\begin{align*}
 \mathcal{D}^{i}_{t,T}(\delta):=\mathrm{e}^{-\gamma_{i}\int_{t}^{T}\delta_{u}^{a,i}\mathbf{1}_{\{\delta_{u}^{i,a}=\underline \delta_{u}^{a}\}}\mathrm{d}N_{u}^{a}+\delta_{u}^{i,b}\mathbf{1}_{\{\delta_{u}^{i,b}=\underline \delta_{u}^{b}\}}\mathrm{d}N_{u}^{b}+Q_{u}^{i}\mathrm{d}S_{u}+\sum_{\ell=1}^K\omega_{\ell}(\delta_{u}^{i,a}\mathbf{1}_{\{\delta_{u}^{i,a}\in K_{\ell}\}}\mathrm{d}N_{u}^{a}+\delta_{u}^{i,b}\mathbf{1}_{\{\delta_{u}^{i,b}\in K_{\ell}\}}\mathrm{d}N_{u}^{b})}.
\end{align*}

First, by the tower property, we have that
\begin{align*}
& V_{t}^{i}(\xi^i,\delta^{-i})
=\underset{\delta^{i}\in \mathcal{A}^{i}(\delta^{-i})}{\text{ess sup}}\mathbb{E}_{t}^{\delta^{i}\otimes_{i}\delta^{-i}}\Big[-\mathcal{D}^{i}_{t,\tau}(\delta)\mathbb{E}_{\tau}^{\delta^{i}\otimes_{i}\delta^{-i}}\big[\mathcal{D}^{i}_{\tau,T}(\delta)\exp\big(-\gamma_{i}\xi^{i}\big)\big]\Big].
\end{align*}
For all $\delta\in \mathcal{A}$, the quotient $\frac{L_{T}^{\delta}}{L_{\tau}^{\delta}}$ does not depend on the value of $\delta$ before time $\tau$. This is by definition of the integrals. Then
\begin{align*}
\mathbb{E}_{\tau}^{\delta^{i}\otimes_{i}\delta^{-i}}\big[\mathcal{D}^{i}_{\tau,T}(\delta)\exp\big(-\gamma_{i}\xi^{i}\big)\big] & =\mathbb{E}^{0}_{\tau}\bigg[-\frac{L_{T}^{\delta^{i}\otimes_{i}\delta^{-i}}}{L_{\tau}^{\delta^{i}\otimes_{i}\delta^{-i}}}\mathcal{D}^{i}_{\tau,T}(\delta)\exp\big(-\gamma_{i}\xi^{i}\big)\bigg] \\
&\leq \underset{\delta^{i}\in \mathcal{A}^{i}(\delta^{-i})}{\text{ess sup}}\mathbb{E}_{\tau}^{\delta^{i}\otimes_{i}\delta^{-i}}\big[-\mathcal{D}^{i}_{\tau,T}(\delta)\exp\big(-\gamma_{i}\xi^{i}\big)\big] \\
&=V_{\tau}^{i}(\xi^i,\delta^{-i}).
\end{align*}
Hence, we obtain that 
\begin{align*}
& V_{t}^{i}(\xi^i,\delta^{-i})\leq \underset{\delta^{i}\in \mathcal{A}^{i}(\delta^{-i})}{\text{ess sup}}\mathbb{E}^{\delta^{i}\otimes_{i}\delta^{-i}}_{t}\big[-V_{\tau}^{i}(\xi^i,\delta^{-i})\mathcal{D}^{i}_{t,\tau}(\delta)\big].
\end{align*}

We next prove the reverse inequality. Let $\delta^{i}\in \mathcal{A}^{i}(\delta^{-i})$ and $\delta^{'i}\in \mathcal{A}^{i}(\delta^{-i})$. We define
\begin{align}\label{Control process proof dynamic programming principle}
(\delta^{i}\otimes_{\tau}\delta^{'i})_{u}:= \delta^{i}_{u}1_{\{0\leq u \leq \tau\}} + \delta^{'i}_{u}1_{\{\tau < u \leq T\}}. 
\end{align}
Then, $\delta^{i}\otimes_{\tau}\delta^{'i}$ being predictable as a sum of two predictable processes,   $\delta^{i}\otimes_{\tau}\delta^{'i}\in \mathcal{A}^{i}(\delta^{-i})$ and 
\begin{align*}
&V_{t}^{i}(\xi^i,\delta^{-i})\geq \mathbb{E}_{t}^{(\delta^{i}\otimes_{\tau}\delta^{'i})\otimes_{i}\delta^{-i}}\big[-\mathcal{D}^{i}_{\tau,T}(\delta')\mathcal{D}^{i}_{t,\tau}(\delta)\exp\big(-\gamma_{i}\xi^{i}\big)\big] \\
& \hspace{4.5em} =\mathbb{E}_{t}^{(\delta^{i}\otimes_{\tau}\delta^{'i})\otimes_{i}\delta^{-i}}\Big[\mathbb{E}_{\tau}^{(\delta^{i}\otimes_{\tau}\delta^{'i})\otimes_{i}\delta^{-i}}\big[-\mathcal{D}^{i}_{\tau,T}(\delta')\exp(-\gamma_{i}\xi^{i})\big]\mathcal{D}^{i}_{t,\tau}(\delta)\Big].
\end{align*}
Using Bayes formula, and noting that $\frac{L_{T}^{(\delta^{i}\otimes_{\tau}\delta^{'i})\otimes_{i}\delta^{-i}}}{L_{\tau}^{(\delta^{i}\otimes_{\tau}\delta^{'i})\otimes_{i}\delta^{-i}}}=\frac{L_{T}^{\delta^{'i}\otimes_{i}\delta^{-i}}}{L_{\tau}^{\delta^{'i}\otimes_{i}\delta^{-i}}}$, we have
\begin{align*}
 \mathbb{E}_{\tau}^{(\delta^{i}\otimes_{\tau}\delta^{'i})\otimes_{i}\delta^{-i}}\big[-\mathcal{D}^{i}_{\tau,T}(\delta')\exp\big(-\gamma_{i}\xi^{i}\big)\big] =\mathbb{E}^{0}_{\tau}\bigg[-\frac{L_{T}^{\delta^{'i}\otimes_{i}\delta^{-i}}}{L_{\tau}^{\delta^{'i}\otimes_{i}\delta^{-i}}}\mathcal{D}^{i}_{\tau,T}(\delta')\exp(-\gamma_{i}\xi^{i})\bigg]=J_{T}^{i}(\xi^i,\tau,\delta^{'i},\delta^{-i}).
\end{align*}

We therefore have 
\begin{align*}
& V_{t}^{i}(\xi^i,\delta^{-i})\geq \mathbb{E}_{t}^{(\delta^{i}\otimes_{\tau}\delta^{'i})\otimes_{i}\delta^{-i}}\big[\mathcal{D}^{i}_{t,\tau}(\delta)J_{T}^{i}(\xi^i,\tau,\delta^{'i},\delta^{-i})\big].
\end{align*}

We can therefore use Bayes's formula and the fact that $\frac{L_{\tau}^{(\delta^{i}\otimes_{\tau}\delta^{'i})\otimes_{i}\delta^{-i}}}{L_{t}^{(\delta^{i}\otimes_{\tau}\delta^{'i})\otimes_{i}\delta^{-i}}}=\frac{L_{\tau}^{\delta^{i}\otimes_{i}\delta^{-i}}}{L_{t}^{\delta^{i}\otimes_{i}\delta^{-i}}}$ to finally obtain 
\begin{align*}
V_{t}^{i}(\xi^i,\delta^{-i})\geq 
&= \mathbb{E}_{t}^{0}\Bigg[\mathbb{E}_{\tau}^{0}\bigg[\frac{L_{T}^{(\delta^{i}\otimes_{\tau}\delta^{'i})\otimes_{i}\delta^{-i}}}{L_{\tau}^{(\delta^{i}\otimes_{\tau}\delta^{'i})\otimes_{i}\delta^{-i}}}\frac{L_{\tau}^{(\delta^{i}\otimes_{\tau}\delta^{'i})\otimes_{i}\delta^{-i}}}{L_{t}^{(\delta^{i}\otimes_{\tau}\delta^{'i})\otimes_{i}\delta^{-i}}}\mathcal{D}^{i}_{t,\tau}(\delta)J_{T}^{i}(\xi^i,\tau,\delta^{'i},\delta^{-i})\bigg]\Bigg] \\ &=\mathbb{E}_{t}^{\delta^{i}\otimes_{i}\delta^{-i}}\big[\mathcal{D}^{i}_{t,\tau}(\delta)J_{T}^{i}(\xi^i,\tau,\delta^{'i},\delta^{-i})\big].
\end{align*}

Since the previous inequality holds for all $\delta^{'i}\in \mathcal{A}^{i}(\delta^{-i})$ we deduce from the monotone convergence theorem together with Lemma \ref{Lemma 6.1} that there exists a sequence $(\delta^{'n})_{n\in \mathbb{N}}$ of controls in $\mathcal{A}^{i}(\delta^{-i})$ such that
\begin{align*}
V_{t}^{i}(\xi^i,\delta^{-i})\geq \lim_{n\rightarrow +\infty}\mathbb{E}_{t}^{\delta^{i}\otimes_{i}\delta^{-i}}\big[\mathcal{D}^{i}_{t,\tau}(\delta)J_{T}^{i}(\xi^i,\tau,\delta^{'n},\delta^{-i})\big] =\mathbb{E}_{t}^{\delta^{i}\otimes_{i}\delta^{-i}}\big[\mathcal{D}^{i}_{t,\tau}(\delta)V_{\tau}^{i}(\xi^i,\delta^{-i})\big],
\end{align*}
thus concluding the proof. 
\end{proof}

\subsection{Proof of Theorem \ref{Theorem market maker}}

We begin with a lemma concerning the integrability of the continuation utility of the $i-$th agent defined in \eqref{Definition continuation utility value function agent}. 
\begin{lemma}\label{Lemma Uniform Integrability}
For all $\delta\in\mathcal{A}$ and all $i\in\{1,\dots,N\}$, the process $V^{i}(\xi^i,\delta^{-i})$ is negative and for a specific $\epsilon>0$, we have
\begin{align*}
    \mathbb{E}^{\delta}\bigg[\sup_{t\in[0,T]}\big|V_{t}^{i}(\xi^i,\delta^{-i})\big|^{1+\epsilon}\bigg]<+\infty, \quad \mathbb{E}^{\delta}\bigg[\sup_{(s,t)\in[0,T]^2}\big(D_{s,t}^{i}(\delta)\big)^{1+\epsilon}\bigg]<+\infty.
\end{align*}
%
\end{lemma}

\begin{proof}
Let $\epsilon >0$, and $\delta\in \mathcal{A}$. Thanks to the uniform boundedness of $\delta^{i}\in \mathcal{A}^{i}(\delta^{-i})$, we have that
\begin{align}\label{Intermediary step for cadlag}
\frac{L_{T}^{\delta^{i}\otimes_{i}\delta^{-i}}}{L_{t}^{\delta^{i}\otimes_{i}\delta^{-i}}}\geq \alpha_{t,T}:= \mathrm{e}^{-\frac{k}{\sigma}(c+\delta_{\infty}(1+H))(N_{T}^{a}-N_{t}^{a}+N_{T}^{b}-N_{t}^{b})-2A\mathrm{e}^{-\frac{kc}{\sigma}}(\mathrm{e}^{\frac{k}{\sigma}(\delta_{\infty}(1+H))}+1)(T-t)}\geq \alpha_{0,T},  
\end{align}
with $H:=\max_{\ell=1,\dots,K}H_{\ell}$. We have 
\begin{align*}
    -V_{t}^{i}(\xi^i,\delta^{-i}) &= \underset{\delta\in \mathcal{A}^{i}(\delta^{-i})}{\text{ess inf}}  \mathbb{E}_t^{\delta^{i}\otimes_{i}\delta^{-i}}\bigg[\exp\bigg(-\gamma_{i}\int_{t}^{T}\delta_{u}^{a}\mathbf{1}_{\{\delta_{u}^{a}=\underline{\delta_u^a\otimes_i\delta_u^{a,-i}}\}}\mathrm{d}N_{u}^{a}+\delta_{u}^{b}\mathbf{1}_{\{\delta_{u}^{b}=\underline{\delta_u^b\otimes_i\delta_u^{b,-i}}\}}\mathrm{d}N_{u}^{b}+Q_{u}^{i}\mathrm{d}S_{u} \\
&\quad +\sum_{\ell=1}^{K}\omega_{\ell}\Big(\delta_{u}^{a}\mathbf{1}_{\{\delta_{u}^{a}\in K_{\ell}\}}\mathrm{d}N_{u}^{a}+\delta_{u}^{b}\mathbf{1}_{\{\delta_{u}^{b}\in K_{\ell}\}}\mathrm{d}N_{u}^{b}\Big)\bigg) \exp\big(-\gamma_{i}\xi^{i}\big) \bigg] \\
& \leq \mathbb{E}_t^{\delta^{i}\otimes_{i}\delta^{-i}}\bigg[\mathrm{e}^{\gamma_i\big(\delta_\infty (1+K\omega_1)(N_T^a+N_T^b)-\int_t^T Q_u^i dS_u\big)}\exp\big(-\gamma_i\xi^i\big)\bigg],
\end{align*}
with $\delta^i\otimes_i \delta^{-i}\in\mathcal{A}$. We used the fact that $N_T^j-N_t^j \leq N_T^j$ for $j\in\{a,b\}$ and for all $i\in\{1,\dots,N\},j\in\{a,b\},t\in[0,T]$,
\begin{align*}
    \exp\bigg(-\gamma_i\int_t^T \delta_u^j \mathbf{1}_{\{\delta_{u}^{a}=\underline{\delta_u^j\otimes_i\delta_u^{j,-i}}\}}\mathrm{d}N_{u}^{j}\bigg) \leq \exp\big(\gamma_i \delta_\infty N_T^j \big).
\end{align*}
Moreover, as $Q^i$ is uniformly bounded by $\overline{q}$, we have for all $L>0$
\begin{align*}
    \mathbb{E}_t^{\delta^{i}\otimes_{i}\delta^{-i}}\Big[\mathrm{e}^{-L\int_t^T Q_u^i dS_u}\Big] \leq \mathrm{e}^{\frac{L^2 \overline{q}^2\sigma^2 T}{2}}. 
\end{align*}
Thus, using Holder's inequality we have
\begin{align*}
    -V_{t}^{i}(\xi^i,\delta^{-i}) &\leq \mathbb{E}_t^{\delta^{i}\otimes_{i}\delta^{-i}}\bigg[\mathrm{e}^{\epsilon\gamma_i\big(\delta_\infty (1+K\omega_1)(N_T^a+N_T^b)-\xi^i\big)}\bigg]^{\frac{1}{\epsilon}}\mathbb{E}_t^{\delta^{i}\otimes_{i}\delta^{-i}}\bigg[\exp\bigg(-(1+\epsilon)\gamma_i\int_t^T Q_u^i dS_u\bigg)\bigg]^{\frac{1}{1+\epsilon}}\\
    & \leq \mathbb{E}_t^{\delta^{i}\otimes_{i}\delta^{-i}}\bigg[\mathrm{e}^{\epsilon\gamma_i\big(\delta_\infty (1+K\omega_1)(N_T^a+N_T^b)-\xi^i\big)}\bigg]^{\frac{1}{\epsilon}} \mathrm{e}^{\frac{(1+\epsilon) \gamma_i^2 \overline{q}^2\sigma^2 T}{2}}. 
\end{align*}
Then, we have
\begin{align*}
     \mathbb{E}^{\delta^{i}\otimes_{i}\delta^{-i}}\bigg[\sup_{t\in[0,T]}\big(-V_{t}^{i}(\xi^i,\delta^{-i})\big)^{1+\epsilon}\bigg] &\leq \mathrm{e}^{\frac{(1+\epsilon)^2 \gamma_i^2 \overline{q}^2\sigma^2 T}{2}}\mathbb{E}^{\delta^{i}\otimes_{i}\delta^{-i}}\bigg[\sup_{t\in[0,T]}\mathbb{E}_t^{\delta^{i}\otimes_{i}\delta^{-i}}\bigg[\mathrm{e}^{\epsilon\gamma_i\big(\delta_\infty (1+K\omega_1)(N_T^a+N_T^b)-\xi^i\big)}\bigg]^{\frac{1+\epsilon}{\epsilon}} \bigg],
\end{align*}
The term inside the conditional expectation is integrable\footnote{Take $\epsilon>1$ together with Condition \eqref{Integrability market maker} for example.} and independent from $t\in[0,T]$ thus by Doob's inequality, we have 
\begin{align*}
    \mathbb{E}^{\delta^{i}\otimes_{i}\delta^{-i}}\bigg[\sup_{t\in[0,T]}\big(-V_{t}^{i}(\xi^i,\delta^{-i})\big)^{1+\epsilon}\bigg] \leq   C \mathrm{e}^{\frac{(1+\epsilon)^2 \gamma_i^2 \overline{q}^2\sigma^2 T}{2}}\mathbb{E}^{\delta^{i}\otimes_{i}\delta^{-i}}\bigg[\mathrm{e}^{\gamma_i^{'}\big(\delta_\infty (1+K\omega_1)(N_T^a+N_T^b)-\xi^i\big)}\bigg],
\end{align*}
where $C>0$ and $\gamma_i^{\prime}=\gamma_i(1+\epsilon)$. Thanks to H\"older's inequality, together with the boundedness of the intensities of the point processes $N^{i,j}$, for $i\in\{1,\dots,N\}$, and $j\in\{a,b\}$, and Condition \eqref{Integrability market maker}, the right--hand side is bounded from above by a term independent of $t\in[0,T]$. The conclusion follows. 

\medskip
Using the same arguments, we have
\begin{align*}
 \mathbb{E}^{\delta}\bigg[\sup_{(s,t)\in[0,T]^2}(\mathcal{D}^{i}_{s,t}(\delta))^{1+\epsilon}\bigg] \leq C^{'}\mathbb{E}^{\delta}\bigg[\mathrm{e}^{\gamma_{i}^{'}\big(\delta_{\infty}(1+K\omega_1)(N_{T}^{a}+N_{T}^{b})+\overline{q}^2\gamma_i\frac{\sigma^2 T}{2}\big)}\bigg]<+\infty
\end{align*}
where $C^{'}>0$, using boundedness of the intensities of the point processes for $i\in\{1,\dots,N\}$. The conclusion follows using Hölder's inequality.
\end{proof}
\medskip
We introduce for all $i\in\{1,\dots,N\}$, and all $\delta\in \mathcal{A}^i(\delta^{-i})$ the process
\begin{align*}
& U_{t}^{\delta\otimes_{i}\delta^{-i}}:=V_{t}^{i}(\xi^i,\delta^{-i})\mathcal{D}_{0,t}^{i}(\delta^{i}\otimes_{i}\delta^{-i}),\; t\in[0,T],
\end{align*}
which thanks to Lemma \ref{Lemma Uniform Integrability} is of class $(D)$.

\medskip
\textbf{Step 1:} Let $\xi\in\mathcal{C}$ be an admissible contract. By definition, there is a Nash equilibrium $\hat\delta(\xi)\in\mathcal{A}$. By use of the dynamic programming principle of Lemma \ref{Lemma 6.2}, for all $\delta^{i}\in\mathcal{A}^{i}(\hat{\delta}^{-i}(\xi))$, the process $U^{\delta^{i}\otimes_i \hat{\delta}^{-i}(\xi)}$ defines a $\mathbb{P}^{\delta^{i}\otimes_i \hat{\delta}^{-i}(\xi)}-$supermartingale. We now check that the process $U^{\hat{\delta}(\xi)}$ is a uniformly integrable $\mathbb{P}^{\hat{\delta}(\xi)}-$martingale. 

\medskip
By Definition \ref{Definition Nash equilibrium}, the control $\hat{\delta}^{i}(\xi)$ is optimal for the $i-$th market maker in the sense that
\begin{align*}
V_{\text{MM}}^{i}\big(\xi^i,\hat{\delta}^{-i}(\xi)\big)\!=\!\mathbb{E}^{\hat\delta(\xi)}\bigg[U_i\bigg(\xi^{i}\!+\!\!\sum_{j\in\{a,b\}}\!\int_{0}^{T}\!\hat\delta_{t}^{i,j}(\xi)\bigg(\mathbf{1}_{\{\hat\delta_{t}^{i,j}(\xi)=\underline {\hat\delta}_{t}^{j}(\xi)\}}\!+\!\sum_{\ell=1}^{K}\int_{0}^{T}\!\omega_{\ell}\mathbf{1}_{\{\hat\delta_{t}^{i,j}(\xi)\in K_{\ell}\}}\bigg)\mathrm{d}N_{t}^{j}\!+\!\int_{0}^{T}Q_{t}^{i}\mathrm{d}S_{t}\bigg) \bigg].
\end{align*}
Hence, an application of the supermartingale property leads, for any $\mathbb{F}$-predictable stopping time $\tau$ taking values in $[0,T]$, to
\begin{align*}
 V_{\text{MM}}^{i}\big(\xi^i,\hat{\delta}^{-i}(\xi)\big)&\geq \mathbb{E}^{\hat\delta(\xi)}\Big[\mathcal{D}_{0,\tau}^{i}\big(\hat\delta(\xi)\big)V_{\tau}^{i}\big(\xi^i,\hat{\delta}^{-i}(\xi)\big)\Big] \\
& \geq \mathbb{E}^{\hat\delta(\xi)}\bigg[U_i\bigg(\xi^{i}\!\!+\!\!\sum_{j\in\{a,b\}}\!\int_{0}^{T}\!\hat\delta_{t}^{i,j}(\xi)\bigg(\mathbf{1}_{\{\hat\delta_{t}^{i,j}(\xi)=\underline {\hat\delta}_{t}^{j}(\xi)\}}\!+\!\sum_{\ell=1}^{K}\!\int_{0}^{T}\!\omega_{\ell}\mathbf{1}_{\{\hat\delta_{t}^{i,j}(\xi)\in K_{\ell}\}}\bigg)\mathrm{d}N_{t}^{j}\!+\!\int_{0}^{T}Q_{t}^{i}\mathrm{d}S_{t}\bigg)\bigg] \\ 
& = V_{\text{MM}}^{i}\big(\xi^i,\hat{\delta}^{-i}(\xi)\big).
\end{align*}
All these inequalities are thus equalities, which proves, since the filtration is right--continuous, that $(U_{t}^{\hat\delta(\xi)})_{t\in[0,T]}$ is a $\mathbb{P}^{\hat\delta(\xi)}-$martingale, and thus for any $t\in[0,T]$
\begin{align*}
U_{t}^{\hat\delta(\xi)}=\mathbb{E}_{t}^{\hat\delta(\xi)}\bigg[U_i\bigg(\xi^{i}+\sum_{j\in\{a,b\}}\int_{0}^{T}\hat\delta_{t}^{i,j}(\xi)\bigg(\mathbf{1}_{\{\hat\delta_{t}^{i,j}(\xi)=\underline {\hat\delta}_{t}^{j}(\xi)\}}+\sum_{\ell=1}^{K}\int_{0}^{T}\omega_{\ell}\mathbf{1}_{\{\hat\delta_{t}^{i,j}(\xi)\in K_{\ell}\}}\bigg)\mathrm{d}N_{t}^{j}+\int_{0}^{T}Q_{t}^{i}\mathrm{d}S_{t}\bigg)\bigg].    
\end{align*}
Using Lemma \ref{Lemma Uniform Integrability}, we conclude that $U^{\hat\delta(\xi)}$ is a uniformly integrable $\mathbb{P}^{\hat\delta(\xi)}-$martingale. Since the filtration $\F$ is right--continuous, we deduce that $U^{\hat\delta(\xi)}$ has a c\`adl\`ag $\mathbb{P}^{\hat\delta(\xi)}-$modification. Since all probability measures here are equivalent, we can assume that $U^{\hat\delta(\xi)}$ actually has c\`adl\`ag paths. As all the probability measures indexed by $\delta\in \mathcal{A}$ are equivalent, we deduce that  $(U_{t}^{\delta})_{t\in[0,T]}$ admits a càdlàg modification, for all $\delta\in\mathcal{A}$.

\medskip
Given the above, for any $\delta^{i}\in\mathcal{A}^{i}(\hat{\delta}^{-i}(\xi))$, we can apply Doob--Meyer's decomposition to the $\mathbb{P}^{\delta^{i}\otimes_i \hat{\delta}^{-i}(\xi)}$ supermartingale of class $(D)$  $U^{\delta^{i}\otimes_i \hat{\delta}^{-i}(\xi)}$ to obtain
\[
U_{t}^{\delta^{i}\otimes_{i}\hat\delta^{-i}(\xi)}=M_{t}^{\delta^{i}\otimes_{i}\hat\delta^{-i}(\xi)}-A_{t}^{\delta^{i}\otimes_{i}\hat\delta^{-i}(\xi),c}-A_{t}^{\delta^{i}\otimes_{i}\hat\delta^{-i}(\xi),d},\; t\in[0,T],
\]
where $M^{\delta^{i}\otimes_{i}\hat\delta^{-i}(\xi)}$ is a uniformly integrable $\mathbb{P}^{\delta^{i}\otimes_{i}\hat\delta^{-i}(\xi)}-$martingale and 
\[
A_t^{\delta^{i}\otimes_{i}\hat\delta^{-i}(\xi)}=A_{t}^{\delta^{i}\otimes_{i}\hat\delta^{-i}(\xi),c}+A_{t}^{\delta^{i}\otimes_{i}\hat\delta^{-i}(\xi),d},\; t\in[0,T],
\]
is an integrable non--decreasing predictable process such that $A_{0}^{\delta^{i}\otimes_{i}\hat\delta^{-i}(\xi),c}=A_{0}^{\delta^{i}\otimes_{i}\hat\delta^{-i}(\xi),d}=0$, with pathwise continuous component $A^{\delta^{i}\otimes_{i}\hat\delta^{-i}(\xi),c}$ and a piecewise constant predictable process $A^{\delta^{i}\otimes_{i}\hat\delta^{-i}(\xi),d}$. 

\medskip
Moreover, using the martingale representation theorem under $\mathbb{P}^{\delta^{i}\otimes_{i}\hat\delta^{-i}(\xi)}$,see \cite[Section A.1]{el2018optimal}, there exist predictable processes
\begin{align*}
\tilde{Z}^{\delta^{i}\otimes_{i}\hat\delta^{-i}(\xi)}=(\tilde{Z}^{\delta^{i}\otimes_{i}\hat\delta^{-i}(\xi),S},\tilde{Z}^{\delta^{i}\otimes_{i}\hat\delta^{-i}(\xi),i,j,a},\tilde{Z}^{\delta^{i}\otimes_{i}\hat\delta^{-i}(\xi),i,j,b}),\;  j\in\{1,\dots,N\},  
\end{align*} 
such that
\begin{align*}
M_{t}^{\delta^{i}\otimes_{i}\hat\delta^{-i}(\xi)}=V_{0}^{i}+\int_{0}^{t}\tilde{Z}_{r}^{\delta^{i}\otimes_{i}\hat\delta^{-i}(\xi),S}\mathrm{d}S_{r}+\sum_{j=1}^{N}\tilde{Z}_{r}^{\delta^{i}\otimes_{i}\hat\delta^{-i}(\xi),i,j,a}\mathrm{d}\tilde{N}_{r}^{\delta^{i}\otimes_{i}\hat\delta^{-i}(\xi),j,a}+\tilde{Z}_{r}^{\delta^{i}\otimes_{i}\hat\delta^{-i}(\xi),i,j,b}\mathrm{d}\tilde{N}_{r}^{\delta^{i}\otimes_{i}\hat\delta^{-i}(\xi),j,b}.
\end{align*}
where the processes $\tilde{N}^{\delta^{i}\otimes_{i}\hat\delta^{-i}(\xi),j,a},\tilde{N}^{\delta^{i}\otimes_{i}\hat\delta^{-i}(\xi),j,b}$, are defined by \eqref{Definition Compensated Poisson Process}.

\medskip
Let $Y^{i}\big(\xi^i,\hat\delta^{-i}(\xi)\big)$ be defined by $V^{i}\big(\xi^i,\hat\delta^{-i}(\xi)\big)=-\mathrm{e}^{-\gamma_{i}Y^{i}(\xi^i,\hat\delta^{-i}(\xi))}$. Since $A^{\delta^{i}\otimes_{i}\hat\delta^{-i}(\xi),d}$ is a predictable point process and the jump times of $(N^{i,a},N^{i,b})$ are totally inaccessible stopping times under $\mathbb{P}^{0}$, we have $\big\langle N^{i,a},A^{\delta^{i}\otimes_i \hat\delta^{-i}(\xi),d}\big\rangle =\big\langle N^{i,b},A^{\delta^{i}\otimes_i \hat\delta^{-i}(\xi),d}\big\rangle =0$, a.s. Using Itô’s formula, we obtain that
\begin{align}\label{Multidimensionnal BSDE}
Y_{T}^{i}\big(\xi^i,\hat\delta^{-i}(\xi)\big)=\xi^{i}, \text{ and } \mathrm{d}Y_{t}^{i}\big(\xi^i,\hat\delta^{-i}(\xi)\big)=\sum_{j=1}^{N}Z_{t}^{i,j,a}\mathrm{d}N_{t}^{j,a}+Z_{t}^{i,j,b}\mathrm{d}N_{t}^{j,b}+Z_{t}^{S,i}\mathrm{d}S_{t}-\mathrm{d}I_{t}^{i}-\mathrm{d}\tilde{A}_{t}^{i,d},
\end{align}
where by direct identification of the coefficients
\begin{align*}
& Z_{t}^{i,j,a}:=-\frac{1}{\gamma_{i}}\log\bigg(1+\frac{\tilde{Z}_{t}^{\delta^{i}\otimes_{i}\hat\delta^{-i}(\xi),i,j,a}}{U_{t^{-}}^{\delta^{i}\otimes_{i}\hat\delta^{-i}(\xi)}}\bigg)-\delta_{t}^{i,a}1_{\{\delta_{t}^{i,a}=\underline{\delta_t^a\otimes_i\hat\delta_t^{a,-i}(\xi)}\}}-\sum_{\ell=1}^{K}\omega_{\ell}\delta_{t}^{i,a}\mathbf{1}_{\{\delta_{t}^{i,a}\in K_{\ell}\}}, \\
& Z_{t}^{i,j,b}:=-\frac{1}{\gamma_{i}}\mathrm{log}\bigg(1+\frac{\tilde{Z}_{t}^{\delta^{i}\otimes_{i}\hat\delta^{-i}(\xi),i,j,b}}{U_{t^{-}}^{\delta^{i}\otimes_{i}\hat\delta^{-i}(\xi)}}\bigg)-\delta_{t}^{i,b}1_{\{\delta_{t}^{i,b}=\underline{\delta_t^b\otimes_i\hat\delta_t^{b,-i}(\xi)}\}}-\sum_{\ell=1}^{K}\omega_{\ell}\delta_{t}^{i,b}\mathbf{1}_{\{\delta_{t}^{i,b}\in K_{\ell}\}},
\end{align*}
\[
Z_{t}^{S,i}:=-\frac{\tilde{Z}_{t}^{\delta^{i}\otimes_{i}\hat\delta^{-i}(\xi),S}}{\gamma_{i}U_{t^{-}}^{\delta^{i}\otimes_{i}\hat\delta^{-i}(\xi)}}-Q_{t^{-}}^{i},\; I_{t}^{i}:=\int_{0}^{t}\bigg(\overline{h}^{i}(\delta_{r}^{i},\hat\delta_{r}^{-i}(\xi),Z_{r}^{i},Q_{r})\mathrm{d}r -\frac{1}{\gamma_{i}U_{r}^{\delta^{i}\otimes_{i}\hat\delta^{-i}(\xi)}}\mathrm{d}A_{r}^{\delta^{i}\otimes_{i}\hat\delta^{-i}(\xi),c}\bigg),
\]
\[
 \overline{h}^{i}(\delta_{t}^{i},\hat\delta_{t}^{-i}(\xi),Z_{t}^{i},Q_{t}):=h(\delta_{t}^{i},\hat\delta_{t}^{-i}(\xi),Z_{t}^{i},Q_{t})-\frac{1}{2}\gamma_{i}\sigma^{2}\big(Z_{t}^{S,i}\big)^{2}, \;  \tilde{A}_{t}^{i,d}:=\frac{1}{\gamma_{i}}\sum_{s\leq t}\mathrm{log}\bigg(1-\frac{\Delta A_{t}^{\delta^{i}\otimes_{i}\hat\delta^{-i}(\xi),d}}{U_{t^{-}}^{\delta^{i}\otimes_{i}\hat\delta^{-i}(\xi)}}\bigg).
\]
In particular, the last relation between $\tilde{A}^{i,d}$ and $A^{\delta^{i}\otimes_{i}\hat\delta^{-i}(\xi),d}$ shows that the process
\begin{align*}
\Delta a_{t}^{i}\!:=\!-\Delta A_{t}^{\delta^{i}\otimes_{i}\hat\delta^{-i}(\xi),d}/U_{t^{-}}^{\delta^{i}\otimes_{i}\hat\delta^{-i}(\xi)}\!\geq\! 0    
\end{align*}
is independent of $\delta^{i}\in \mathcal{A}^{i}(\hat\delta^{-i}(\xi))$. 

\medskip
We now prove that, $A^{\delta^{i}\otimes_{i}\hat\delta^{-i}(\xi),d}\!=\!-\sum_{0<s\leq \cdot}\!U_{s^{-}}^{\delta^{i}\otimes_{i}\hat\delta^{-i}(\xi)}\!\Delta a_{s}^{i}=0$ so that
\begin{align*}
\tilde{A}^{i,d}\!=\!0, \quad   I_{t}\!=\!\int_{0}^{\cdot}\overline{H}^{i}(\hat\delta^{-i}(\xi),Z_{r}^{i},Q_{r})\mathrm{d}r,
\end{align*}
where 
\begin{align*}
\overline{H}^{i}(\hat\delta^{-i}(\xi),Z_{t}^{i},Q_{t})=H^{i}(\hat\delta^{-i}(\xi),Z_{t}^{i},Q_{t})-\frac{1}{2}\gamma_{i}\sigma^{2}(Z_{t}^{S,i})^{2}.
\end{align*}
As $V_{T}^{i}\big(\xi^i,\hat\delta^{-i}(\xi)\big)=-1$, note that
\begin{align*}
 0&=\sup_{\delta^{i}\in \mathcal{A}^{i}(\hat\delta^{-i}(\xi))}\mathbb{E}^{\delta^{i}\otimes_{i}\hat\delta^{-i}(\xi)}\Big[U_{T}^{\delta^{i}\otimes_{i}\hat\delta^{-i}(\xi)}\Big]-V_{0}^{i}(\xi^i,\delta^{-i})\\
&=\sup_{\delta^{i}\in \mathcal{A}^{i}(\hat\delta^{-i}(\xi))}\mathbb{E}^{\delta^{i}\otimes_{i}\hat\delta^{-i}(\xi)}\Big[U_{T}^{\delta^{i}\otimes_{i}\hat\delta^{-i}(\xi)}-M_{T}^{\delta^{i}\otimes_{i}\hat\delta^{-i}(\xi)}\Big]\\
& =\gamma_{i}\sup_{\delta^{i}\in \mathcal{A}^{i}(\hat\delta^{-i}(\xi))}\mathbb{E}^{0}\bigg[L_{T}^{\delta^{i}\otimes_{i}\hat\delta^{-i}(\xi)}\int_{0}^{T}U_{r^{-}}^{\delta^{i}\otimes_{i}\hat\delta^{-i}(\xi)}\bigg(\mathrm{d}I_{r}^{i}-\overline{h}^{i}(\delta_{r}^{i},\hat\delta_{r}^{-i}(\xi),Z_{r}^{i},Q_{r})\mathrm{d}r+\frac{\mathrm{d}a_{r}^{i}}{\gamma_{i}}\bigg)\bigg].  
\end{align*}
Moreover, since the controls are uniformly bounded, we have by Lemma \ref{Lemma Uniform Integrability}
\begin{align*}
U_{t}^{\delta^{i}\otimes_{i}\hat\delta^{-i}(\xi)}\leq -\beta_{t}^{i}=V_{t}^{i}\big(\xi^i,\hat\delta^{-i}(\xi)\big)\mathrm{e}^{-2\delta_{\infty}(N_{T}^{a}-N_{0}^{a}+N_{T}^{b}-N_{0}^{b})-\gamma_{i}\int_{0}^{t}Q_{r}^{i}\mathrm{d}S_{r}}<0.
\end{align*}
Since $A^{\delta^{i}\otimes_{i}\hat\delta^{-i}(\xi),d}\geq 0,$ $U^{\delta^{i}\otimes_{i}\hat\delta^{-i}(\xi)}\leq 0$, and $\mathrm{d}I_{t}^{i}-\overline{h}^{i}(\delta_{t}^{i},\hat\delta_{t}^{-i}(\xi),Z_{t}^{i},Q_{t})\mathrm{d}t\geq 0$, we obtain 
\begin{align*}
 0&\leq \sup_{\delta^{i}\in \mathcal{A}^{i}(\hat\delta^{-i}(\xi))}\mathbb{E}^{0}\bigg[\alpha_{0,T}\int_{0}^{T}-\beta_{r^{-}}^{i}\bigg(\mathrm{d}I_{r}^{i}-\overline{h}^{i}(\delta_{r}^{i},\hat\delta_{r}^{-i}(\xi),Z_{r}^{i},Q_{r})\mathrm{d}r + \frac{\mathrm{d}a_{r}^{i}}{\gamma_{i}}\bigg)\bigg]\\
&  =-\mathbb{E}^{0}\bigg[\alpha_{0,T}\int_{0}^{T}\beta_{r^{-}}^{i}\bigg(\mathrm{d}I_{r}^{i}-\overline{H}^{i}(\hat\delta_{r}^{-i}(\xi),Z_{r}^{i},Q_{r})\mathrm{d}r + \frac{\mathrm{d}a_{r}^{i}}{\gamma_{i}}\bigg)\bigg].
\end{align*}

The quantities $\alpha_{0,T}\int_{0}^{T}\beta_{r^{-}}^{i}\big(\mathrm{d}I_{r}^{i}-\overline{H}^{i}(\hat\delta_{r}^{-i}(\xi),Z_{r}^{i},Q_{r})\big)\mathrm{d}r$ and $\alpha_{0,T}\int_{0}^{T}\beta_{r^{-}}^{i}\frac{\mathrm{d}a_{r}^{i}}{\gamma_{i}}$ being non--negative random variables, this implies the announced result. 

\medskip
Given the dynamic under $\mathbb{P}^{\hat\delta(\xi)}$ of the process $U_{t}^{\hat\delta(\xi)}$, Itô's formula leads to
\begin{align*}
 \mathrm{d}U_{t}^{\delta^{i}\otimes_{i}\hat\delta^{-i}(\xi)}=&\ \sum_{j=1}^{N}\tilde{Z}_{t}^{\hat\delta(\xi),i,j,a}\mathrm{d}\tilde{N}_{t}^{\delta^{i}\otimes_{i}\hat\delta^{-i}(\xi),j,a}+\tilde{Z}_{t}^{\hat\delta(\xi),i,j,b}\mathrm{d}\tilde{N}_{t}^{\delta^{i}\otimes_{i}\hat\delta^{-i}(\xi),j,b}+\tilde{Z}_{t}^{\hat\delta(\xi),i,S}\mathrm{d}S_{t}\\
& +U_{t}^{\delta^{i}\otimes_{i}\hat\delta^{-i}(\xi)}\Big(h\big(\hat\delta_{t}(\xi),\hat\delta_{t}^{-i}(\xi),Z_{t}^{\hat\delta(\xi),i},Q_{t}\big)-h\big(\delta^{i}_{t},\hat\delta_{t}^{-i}(\xi),Z_{t}^{\hat\delta(\xi),i},Q_{t}\big) \Big)\mathrm{d}t
\end{align*}
and the $\mathbb{P}^{\delta^{i}\otimes_{i}\delta^{*-i}}-$supermartingale property implies that, almost surely for all $t\in[0,T]$
\begin{align*}
h\big(\hat\delta_{t}^{i}(\xi),\hat\delta^{-i}(\xi),Z_{t}^{\hat\delta(\xi),i},Q_{t}\big)-h\big(\delta_{t}^{i},\hat\delta_{t}^{-i}(\xi),Z_{t}^{\hat\delta(\xi),i},Q_{t}\big) \geq 0.
\end{align*}
Hence
\begin{align*}
\hat\delta_{t}^{i}(\xi)\in \underset{\delta\in \mathcal{B}_{\infty}^{2}}{\textup{argmax }}h\big(\delta,\hat\delta_{t}^{-i}(\xi),\tilde{Z}^{\hat\delta(\xi),i}_t,Q_t\big). 
\end{align*}
Finally, we check that $Z\in \mathcal{Z}$. 
Using Lemma \ref{Lemma Uniform Integrability}, we have that
\begin{align*}
\sup_{\delta^{i}\in \mathcal{A}^{i}(\hat\delta^{-i}(\xi))}\mathbb{E}^{\delta^{i}\otimes_{i}\hat\delta^{-i}(\xi)}\Big[\sup_{t\in[0,T]}|U_{t}^{\delta^{i}\otimes_{i}\hat\delta^{-i}(\xi)}|^{p^\prime+1}\Big]<+\infty   
\end{align*}
for some $p^\prime>0$. The desired conclusion comes from the fact that
\begin{align*}
& \mathrm{e}^{-\gamma_{i}Y^{i}_{t}}=U^{\delta^{i}\otimes_{i}\hat\delta^{-i}(\xi)}_{t}\mathcal{D}_{0,t}\big(\delta^{i}\otimes_i \hat\delta^{-i}(\xi)\big).
\end{align*}

\begin{remark}
Note that we described here a solution to the following system of $N$ {\rm BSDEs} given by, for all $i\in\{1,\dots,N\}$
\begin{align*}
Y_{t}^{i,y_{0},Z,\hat\delta}\!:=\xi^{i}\!-\!\sum_{j=1}^{N}\int_{t}^{T}\! Z_{r}^{i,j,a}\mathrm{d}N_{r}^{j,a}\!+\!Z_{r}^{i,j,b}\mathrm{d}N_{r}^{j,b}\!+\!Z_{r}^{S,i}\mathrm{d}S_{r}\!+\!\bigg(\frac{1}{2}\gamma_{i}\sigma^{2}(Z_{r}^{S,i}+Q_{r}^{i})^{2}-H^{i}(\hat\delta^{-i}(\xi),Z_{r}^{i},Q_{r}^{i})\bigg)\mathrm{d}r. 
\end{align*}
\end{remark}

\medskip
\textbf{Step 2:} Conversely, let us be given a contract vector $\xi=Y_{T}^{y_0,Z,\hat\delta}\in\Xi$, with $(Y_0,Z)\in\mathbb{R}^N \times \mathcal{Z}$ and $\hat\delta\in\Oc$. For $i\in\{1,\dots,N\}$, we note
\begin{align*}
V_{t}^{i}\big(Y_{T}^{i,y_0,Z,\hat\delta},\hat\delta^{-i}(Y_{T}^{y_0,Z,\hat\delta})\big):=-\mathrm{e}^{-\gamma_{i}Y_{t}^{i,y_0,Z,\hat\delta}}.
\end{align*}
Given an arbitrary bid--ask policy $\delta^{i}\in \mathcal{A}^{i}(\hat\delta^{-i})$ of the $i-$th agent, an application of Itô's formula leads to
\begin{align*}
 \mathrm{d}U_{t}^{\delta^{i}\otimes_{i}\hat\delta^{-i}}=&\ -\gamma_{i}U_{t}^{\delta^{i}\otimes_{i}\hat\delta^{-i}}\bigg((Q_{t}^{i}+Z_{t}^{S,i})\mathrm{d}S_{t}-(H^{i}(\hat\delta_{t}^{-i},Z^{i}_{t},Q_{t})-h^{i}(\delta_{t}^{i},\hat\delta_{t}^{-i},Z_{t}^{i},Q_{t}))\mathrm{d}t\\
& +\gamma_{i}^{-1}\Big(1-\exp\big(-\gamma_{i}\big(Z_{t}^{i,a}+\delta_{t}^{i,a}1_{\{\delta_{t}^{i,a}=\underline{\delta_t^a\otimes_i\hat\delta_t^{a,-i}}\}}+\sum_{\ell=1}^{K}\omega_{\ell}\delta_{t}^{i,a}\mathbf{1}_{\{\delta_{t}^{i,a}\in K_{\ell}\}}\big)\big)\Big)\mathrm{d}\tilde{N}_{t}^{\delta^{i}\otimes_i \hat\delta^{-i},a}\\
& +\gamma_{i}^{-1}\Big(1-\exp\big(-\gamma_{i}\big(Z_{t}^{i,b}+\delta_{t}^{i,b}1_{\{\delta_{t}^{i,b}=\underline{\delta_t^b\otimes_i\hat\delta_t^{b,-i}}\}}+\sum_{\ell=1}^{K}\omega_{\ell}\delta_{t}^{i,b}\mathbf{1}_{\{\delta_{t}^{i,b}\in K_{\ell}\}}\big)\big)\Big)\mathrm{d}\tilde{N}_{t}^{\delta^{i}\otimes_i \hat\delta^{-i},b}\bigg).
\end{align*}
Hence, $\big(U_{t}^{\delta^{i}\otimes_{i}\hat\delta^{-i}}\big)_{t\in[0,T]}$ is a $\mathbb{P}^{\delta^{i}\otimes_{i}\hat\delta^{-i}}-$local supermartingale.  Thanks to Lemma \ref{Lemma Uniform Integrability}, $\big(U_{t}^{\delta^{i}\otimes_{i}\hat\delta^{-i}}\big)_{t\in[0,T]}$ is of class $(D)$ and therefore is a true supermartingale. We obtain that
\begin{align*}
&-\int_{0}^{\cdot}\gamma_{i}U_{t}^{\delta^{i}\otimes_{i}\hat\delta^{-i}}\bigg((Q_{t}^{i}+Z_{t}^{S,i})\mathrm{d}S_{t}\\
& +\gamma_{i}^{-1}\Big(1-\exp\big(-\gamma_{i}(Z_{t}^{i,a}+\delta^{i,a}_{t}1_{\{\delta_{t}^{i,a}=\underline{\delta_t^a\otimes_i\hat\delta_t^{a,-i}}\}}+\sum_{\ell=1}^{K}\omega_{\ell}\delta_{t}^{i,a}\mathbf{1}_{\{\delta_{t}^{i,a}\in K_{\ell}\}})\big)\Big)\mathrm{d}\tilde{N}_{t}^{\delta^{i}\otimes_i \hat\delta^{-i},a}\\
& +\gamma_{i}^{-1}\Big(1-\exp\big(-\gamma_{i}(Z_{t}^{i,b}+\delta_{t}^{i,b}1_{\{\delta_{t}^{i,b}=\underline{\delta_t^b\otimes_i\hat\delta_t^{b,-i}}\}}+\sum_{\ell=1}^{K}\int_{0}^{T}\omega_{\ell}\delta_{t}^{i,b}\mathbf{1}_{\{\delta_{t}^{i,b}\in K_{\ell}\}})\big)\Big)\mathrm{d}\tilde{N}_{t}^{\delta^{i}\otimes_i \hat\delta^{-i},b}\bigg)
\end{align*}
is a true martingale. Therefore
\begin{align*}
J_{\text{MM}}^{i}(\xi^{i},\delta^{i},\hat\delta^{-i})&=\mathbb{E}^{\delta^{i}\otimes_{i}\delta^{-i}}\Big[U_{T}^{\delta^{i}\otimes_{i}\hat\delta^{-i}}\Big]\\
& =-\mathrm{e}^{-\gamma_{i}y_{0}^{i}}+\mathbb{E}^{\delta^{i}\otimes_{i}\delta^{-i}}\bigg[\int_{0}^{T}\gamma_{i}U_{t}^{\delta^{i}\otimes_{i}\hat\delta^{-i}}\big(H^{i}(\hat\delta_{t}^{-i},Z^{i}_{t},Q_{t})-h^{i}(\delta_{t}^{i},\hat\delta_{t}^{-i},Z_{t}^{i},Q_{t})\big)\mathrm{d}t\bigg] \leq -\mathrm{e}^{-\gamma_{i}y_{0}^{i}}.
\end{align*}
In addition to this, the previous inequality becomes an equality if and only if $\delta^{i}$ is chosen as the maximiser of the Hamiltonian $h^{i}$. By definition, it means that $U^{\hat\delta}$ is a $\mathbb{P}^{\hat\delta}-$martingale and that $\hat\delta^{i}$ is the optimal control for the $i-$th agent, in the sense of \eqref{Definition Nash equilibrium}. As this property holds for any $i\in\{1,\dots,N\}$, it means that $\hat\delta$ is a Nash equilibrium.

\medskip
Finally as we showed that the contracts in $\Xi$ generates at least one Nash equilibrium we have the inclusion $\mathcal{C} \supset \Xi$. Hence, the equality $\Xi=\mathcal{C}$ is proved.

\subsection{Proof of Lemma \ref{Lemma unicity fixed point Hamiltonian}}

For $(z,q)\in\mathcal{R}^N \times \mathbb{Z}^N$, we set $z^{i,\ell,j}=z^{i,j}$ for all $(i,\ell)\in \{1,\dots,N\}^2$ and $j\in\{a,b\}$. Hence, the Hamiltonian of the $i-$th agent reduces to
\[
h^{i}(d^{i},d^{-i},z^{i},q):=\sum_{j\in\{a,b\}}\gamma_{i}^{-1}\bigg(1-\exp\bigg(-\gamma_{i}\bigg(z^{i,j}+d^{i,j}\mathbf{1}_{\{d^{i,j}=\underline{d^j\otimes_i d^{j,-i}}\}}+\sum_{k=1}^{K}\omega_{k}d^{i,j}\mathbf{1}_{\{d^{i,j}\in K_{k}\}} \bigg)\bigg)\bigg)\lambda^{j}(d^{j},q).
\]

For $i\in\{1,\dots,N\}$, an optimisation of $h^{i}(d^{i},d^{-i},z^{i},q)$ with respect to $d^i$ leads to a unique\footnote{Uniqueness follows from strict concavity of the vector $h^i$ with respect to $d\in\mathcal{B}_{\infty}^{2N}$} maximum defined as $d^{\star i,j}(z,q)=\Delta^{i,j}(z,q)$ for $i\in\{1,\dots,N\}$, and $j\in\{a,b\}$. This maximiser completely characterise the behaviour of the $i-$th agent compared to the position of the $N-1$ others. 

\medskip
Moreover, no matter if the $i-$th agent plays the best spread or not, compared to the response of the other agents, his optimal response will lead to the following value
\begin{align*}
h^{i}(\Delta^{i,:}(z,q),\Delta^{-i,:}(z,q),z^{i},q)=\frac{\sigma}{1+\frac{\sigma\gamma_i}{k\varpi}}\lambda^{j}\big(\Delta^{:,j}(z,q),q\big), \; i\in \{1,\dots,N\}.
\end{align*}
Hence, when the $N$ agents play $\Delta$, they have no interest in switching their bid--ask policy. Thus, it characterises a unique fixed point of the Hamiltonian.  

\subsection{Exchange's Hamiltonian maximisation}\label{Section Maximization hamiltonian}

The following technical result follows from direct but tedious computations. It provides condition on $\delta_\infty$ under which the maximisers defined in \eqref{Lemma maximizers HJB} exist. 

\begin{lemma}\label{Lemma exchange Hamiltonian}
Let $q\in\mathcal{Q}^N$, $c\in \mathbb{R}$, $(\eta,k,\sigma)\in(0,+\infty)^3,$ $\gamma_{i}>0$ for all $i\in\{1,\dots,N\}$, and $v_{0},\dots,v_{N} < 0$. Then, for $z\in \mathcal{R}^{N}$ and $i\in\{1,\dots,N\},$ $j\in\{a,b\}$, we define
\begin{align*}
& \Phi_{q}^{i,j}(z):=\lambda^{i,j}\big(\Delta^{:,j}(z,q),q\big) \bigg(\mathrm{e}^{\eta(Nz^{j}-c)}v_{i}-v_{0}\mathcal{L}^{j}\big(\Delta(z,q)\big) \Big)\bigg),\; \Phi^{j}_{q}(z):=\sum_{i=1}^{N}\Phi_{q}^{i,j}(z),
\end{align*}
with $\Delta(z,q)$ defined as in {\rm Lemma \ref{Corollary best response agent}}, and $\delta_{\infty}>0$. Assume that
\begin{align*}
\delta_{\infty}\geq C_{\infty}+\frac{N}{\eta}\bigg|\mathrm{log}\bigg(\frac{v_{0}}{\sum_{i=1}^{N}v_{i}}\bigg)\bigg|,
\end{align*}
with $C_{\infty}:=N|c|+\sum_{i=1}^{N}\bigg(\Big(\frac{1}{\eta}+\frac{1}{\gamma_{i}}\Big)\mathrm{log}\Big(1+\frac{\sigma\gamma_{i}}{k\varpi}\Big)\bigg)-\frac{N}{\eta}\mathrm{log}\Big(\frac{k\varpi}{k\varpi+\sigma\eta}\big(1+\eta\sigma\sum_{i=1}^{N}\frac{1}{k\varpi+\sigma\gamma_{i}}\big)\Big)$. Then, the functions $\Phi_{q}^{j}$, $j\in\{a,b\}$, admit a maximum $z^{\star}$ given by 
\begin{align*}
&z^{\star}:=\frac{1}{N}\Bigg(c+\frac{1}{\eta}\mathrm{log}\bigg(\frac{v_{0}}{\sum_{i\in \mathcal{G}}v_{i}}\bigg) +\frac{1}{\eta}\mathrm{log}\bigg(\frac{k\varpi}{k\varpi+\sigma\eta}{\rm Card }(\mathcal{G})\bigg(1+\eta\sigma\sum_{i=1}^{N}\frac{1}{k\varpi+\sigma\gamma_{i}}\bigg)\bigg)\Bigg).
\end{align*}
Moreover
\begin{align*}
\Phi^{j}_{q}(z^{\star})=-C v_{0}\exp\bigg(\frac{k\varpi}{\sigma\eta}\mathrm{log}\Big(\frac{v_{0}}{\sum_{i\in \mathcal{G}}v_{i}}\Big)\bigg),
\end{align*}
where 
\begin{align*}
C:=&\ A\exp\bigg(-\frac{k}{\sigma}\Big(c\big(1-\varpi\big)-\frac{\varpi}{\eta}\mathrm{log}\bigg(\frac{k\varpi}{k\varpi+\eta\sigma}{\rm Card }(\mathcal{G})\bigg(1+\eta\sigma\sum_{i=1}^{N}\frac{1}{k\varpi+\sigma\gamma_{i}}\bigg)\bigg)\\
& +\varpi\sum_{i=1}^{N}\gamma_{i}^{-1}\mathrm{log}\bigg(1+\frac{\sigma\gamma_{i}}{k\varpi}\bigg)\bigg)\bigg)\times \frac{\sigma\eta}{k\varpi+\sigma\eta}\bigg(1+\eta\sigma\sum_{i=1}^{N}\frac{1}{k\varpi+\sigma\gamma_{i}}\bigg).
\end{align*}
\end{lemma}

\subsection{Proof of Lemma \ref{Lemma Cauchy Lipschitz}}
As the state variables $q^{i},i\in\{1,\dots,N\}$ live in a discrete compact set, PDE \eqref{PDE non-linear before change of variable} is in fact a system of $(2\overline{q}+1)^{N}$ ordinary differential equations. Hence, a use of Cauchy--Lipschitz theorem will provide existence and unicity. We define
\begin{align*}
\mathcal{S}:=\{x\in \mathbb{R}: l_b < x < u_b < 0,\; (l_b,u_b)\in \mathbb{R}^{\star 2}_{-}\}.
\end{align*}
Fix some subsets $I$, and $J$ of $\{1,\dots,N\}$, as well as vectors $(x_{\oplus},x_{\ominus},x)\in \mathbb{R}^{{\rm Card}(I)}\times\mathbb{R}^{{\rm Card}(J)}\times \mathcal{S}$. Then, we introduce for $q\in \mathcal{Q}^N$, the map
\begin{align*}
T_{q}(x,x_{\oplus},x_{\ominus})=-xC^{S}(q)+xC\Bigg(\bigg(\frac{x}{\sum_{j\in J}x^{i}_{\ominus}}\bigg)^{\frac{k\varpi}{\sigma\eta}}+\bigg(\frac{x}{\sum_{i\in I}x^{i}_{\oplus}}\bigg)^{\frac{k\varpi}{\sigma\eta}}\Bigg),
\end{align*}
where $C^{S}(q)$,and $C$ come from \eqref{PDE non-linear before change of variable}. We now show that this application is Lipschitz. Direct computations show that for any $(i,j)\in I\times J$
\begin{align*}
& \partial_{x}T_{q}=-C^{S}(q)+\Big(1+\frac{k\varpi}{\sigma\eta}\Big)C\bigg(\bigg(\frac{x}{\sum_{j\in J}x^{j}_{\ominus}}\bigg)^{\frac{k\varpi}{\sigma\eta}}+\bigg(\frac{x}{\sum_{i\in I}x^{i}_{\oplus}}\bigg)^{\frac{k\varpi}{\sigma\eta}}\bigg), \\
& \partial_{x^{i}_{\oplus}}T_{q}=-\frac{k\varpi}{\sigma\eta}C\bigg(\frac{x}{\sum_{i\in I}x^{i}_{\oplus}}\bigg)^{1+\frac{k\varpi}{\sigma\eta}}, \; \partial_{x^{j}_{\ominus}}T_{q}=-\frac{k\varpi}{\sigma\eta}C\bigg(\frac{x}{\sum_{j\in J}x^{j}_{\ominus}}\bigg)^{1+\frac{k\varpi}{\sigma\eta}} .
\end{align*}
By the fact that $(x,x_{\oplus},x_{\ominus})\in \mathcal{S}\times\mathcal{S}^{\#I}\times\mathcal{S}^{\#J}$, the gradient of $T_q$ is uniformly bounded (in the $\|\cdot\|_{\infty}$ sense).

\subsection{Proof of Theorem \ref{Theorem Verification}}\label{Section Verification Argument appendix}
We begin this section with a technical lemma.
\begin{lemma} \label{Lemma Verification argument}
Let $Z\in \mathcal{Z}$, and define $\xi:=Y_{T}^{0,Z,\Delta(Z,Q)}$. We define
\begin{align*}
K_{t}^{Z}:=\exp\Big(-\eta\big(c(N_{t}^{a}+N_{t}^{b})-Y_{t}^{0,Z,\Delta}\cdot 1_N \big)\Big),\; t\in [0,T].
\end{align*} 
There exists $C>0,$ and $\epsilon>0$ such that
\begin{align*}
\mathbb{E}^{\Delta(Z,Q)}\bigg[\sup_{t\in [0,T]}|K_{t}^{Z}|^{1+\epsilon}\bigg]\leq C.
\end{align*}
\end{lemma}
\begin{proof}
We define for all $i\in\{1,\dots,N\}$ the processes
\begin{align*}
    \overline{Y}_t^{i,Z,\Delta} & := Y_t^{i,0,Z,\Delta} + \int_{0}^{t}\delta_{u}^{a,i}\mathbf{1}_{\{\delta_{u}^{i,a}=\underline \delta_{u}^{a}\}}\mathrm{d}N_{u}^{a}+\delta_{u}^{i,b}\mathbf{1}_{\{\delta_{u}^{i,b}=\underline \delta_{u}^{b}\}}\mathrm{d}N_{u}^{b}+Q_{u}^{i}\mathrm{d}S_{u}\\
    & \quad +\sum_{\ell=1}^K\omega_{\ell}(\delta_{u}^{i,a}\mathbf{1}_{\{\delta_{u}^{i,a}\in K_{\ell}\}}\mathrm{d}N_{u}^{a}+\delta_{u}^{i,b}\mathbf{1}_{\{\delta_{u}^{i,b}\in K_{\ell}\}}\mathrm{d}N_{u}^{b}),
\end{align*}
and we rewrite
\begin{align*}
    K_t^Z = \exp\bigg(\eta\bigg(\sum_{i=1}^N \overline{Y}_t^{i,Z,\Delta} \bigg)\bigg)\exp\Bigg(&-\eta\bigg(c(N_t^a+N_t^b) + \sum_{i=1}^N\bigg( \int_{0}^{t}\delta_{u}^{a,i}\mathbf{1}_{\{\delta_{u}^{i,a} =\underline \delta_{u}^{a}\}}\mathrm{d}N_{u}^{a}+\delta_{u}^{i,b}\mathbf{1}_{\{\delta_{u}^{i,b}=\underline \delta_{u}^{b}\}}\mathrm{d}N_{u}^{b}+Q_{u}^{i}\mathrm{d}S_{u}\\
    & \quad +\sum_{\ell=1}^K\omega_{\ell}(\delta_{u}^{i,a}\mathbf{1}_{\{\delta_{u}^{i,a}\in K_{\ell}\}}\mathrm{d}N_{u}^{a}+\delta_{u}^{i,b}\mathbf{1}_{\{\delta_{u}^{i,b}\in K_{\ell}\}}\mathrm{d}N_{u}^{b})\bigg)\bigg)\Bigg).
\end{align*}
Using Step 2 of the proof of Theorem \ref{Theorem market maker}, we know what for all $i\in \{1,\dots,N\}$, $\mathrm{e}^{-\gamma_i \overline{Y}_t^{i,Z,\Delta}}$ is a $\mathbb{P}^{\hat{\delta}}$-martingale. Using Jensen's inequality with the convex function (on $\mathbb{R}_+^{\star}$) $\phi_i(x)=x^{-\frac{\eta}{\gamma_i}}$, and condition \eqref{Integrability market maker} with $\xi^i = \overline{Y}_T^{i,Z,\Delta}$,  we have
\begin{align*}
    \mathbb{E}_t^{\hat{\delta}}\Big[\mathrm{e}^{\eta \overline{Y}_T^{i,Z,\Delta}}\Big] =  \mathbb{E}_t^{\hat{\delta}}\Big[\phi_i\Big(\mathrm{e}^{-\gamma_i \overline{Y}_T^{i,Z,\Delta}}\Big)\Big]\geq \phi_i\Big(\mathrm{e}^{-\gamma_i \overline{Y}_t^{i,Z,\Delta}}\Big) = \mathrm{e}^{\eta \overline{Y}_t^{i,Z,\Delta}}. 
\end{align*}
Similar computations show that $\Big(\mathrm{e}^{\eta \overline{Y}_t^{i,Z,\Delta}}\Big)_{t\in[0,T]}$ is a positive $\mathbb{P}^{\hat{\delta}}$-submartingale. By using Jensen's inequality, we have 
\begin{align*}
    \mathbb{E}^{\Delta(Z,Q)}\bigg[\sup_{t\in [0,T]}\exp\bigg(\eta^{'}\bigg(\sum_{i=1}^N \overline{Y}_t^{i,Z,\Delta} \bigg)\bigg)\bigg] \leq \frac{1}{N}\sum_{i=1}^N \mathbb{E}^{\Delta(Z,Q)}\bigg[\sup_{t\in [0,T]}\exp\bigg(N\eta^{'} \overline{Y}_t^{i,Z,\Delta}\bigg)\bigg],
\end{align*}
where $\eta' = \eta(1+\epsilon)$. Using Doob's inequality, there exists positive constants $k^i_1$ independent from $t\in[0,T]$ such that for all $i\in\{1,\dots,N\}$,
\begin{align*}
    \mathbb{E}^{\Delta(Z,Q)}\bigg[\sup_{t\in [0,T]}\exp\bigg(N\eta^{'} \overline{Y}_t^{i,Z,\Delta}\bigg)\bigg]\leq k_1^i \mathbb{E}^{\Delta(Z,Q)}\bigg[\exp\bigg(N\eta^{'} \overline{Y}_T^{i,Z,\Delta}\bigg)\bigg]. 
\end{align*}
Using Holder's inequality, and noting that 
\begin{align*}
    &\exp\bigg(N\eta^{'}\bigg(\int_{0}^{t}\delta_{u}^{a,i}\mathbf{1}_{\{\delta_{u}^{i,a}=\underline \delta_{u}^{a}\}}\mathrm{d}N_{u}^{a}+\delta_{u}^{i,b}\mathbf{1}_{\{\delta_{u}^{i,b}=\underline \delta_{u}^{b}\}}\mathrm{d}N_{u}^{b}+Q_{u}^{i}\mathrm{d}S_{u}\\
    & \quad +\sum_{\ell=1}^K\omega_{\ell}(\delta_{u}^{i,a}\mathbf{1}_{\{\delta_{u}^{i,a}\in K_{\ell}\}}\mathrm{d}N_{u}^{a}+\delta_{u}^{i,b}\mathbf{1}_{\{\delta_{u}^{i,b}\in K_{\ell}\}}\mathrm{d}N_{u}^{b})\bigg)\bigg),
\end{align*}
has moments of all orders by boundedness of the intensities of $N^a,N^b$, there exists $\eta^{''}>\eta^{'}$ and a constant $k_2^i>0$ such that
\begin{align*}
  k_1^i \mathbb{E}^{\Delta(Z,Q)}\bigg[\exp\bigg(N\eta^{'} \overline{Y}_T^{i,Z,\Delta}\bigg)\bigg] \leq k_2^i   \mathbb{E}^{\Delta(Z,Q)}\bigg[\exp\bigg(N\eta^{''} \overline{Y}_T^{i,Z,\Delta}\bigg)\bigg] <+\infty,
\end{align*}
where we used Condition \eqref{Integrability principal} with $\xi^i = \overline{Y}_T^{i,Z,\Delta}$. The conclusion follows using again the boundedness of the intensities of point processes in the definition of $K^Z$. 
\end{proof}
To prove Theorem \ref{Theorem Verification}, we verify that the function $v$ introduced in \eqref{HJB equation} coincides at $\big(0,Q_{0}\big)$ with the value function of the reduced exchange problem with maximum achieved at the optimum $z^{*}(t,Q_t)$ in \eqref{Lemma maximizers HJB}. 

\medskip
The function $v$ is negative bounded Moreover, since $\delta_{\infty}\geq \Delta_{\infty}$, it follows that $v$ is a solution of \eqref{HJB equation}. A direct application of Itô’s formula coupled with substitution of \eqref{HJB equation} leads to
\begin{align}\label{Intermediary step uniform integrability verification}
\frac{\mathrm{d}(v(t,Q_{t})K_{t}^{Z})}{K_{t^{-}}^{Z}}\!=\!\big(h_{t}^{Z}\!-\!\mathcal{H}_{t}\big)\mathrm{d}t \!+\! \eta \! \sum_{i=1}^{N}\!\bigg(\!v\big(t,Q_{t}\big)Z_{t}^{S,i}\mathrm{d}S_{t}\!+\!\!\!\sum_{j\in\{a,b\}}^{}\!\!\!\Big(\!v\big(t,Q_{t^{-}}^{i}\!+\!\!\!\phi(j)\big)\mathrm{e}^{\eta(NZ^{j}-c)}\!\!-v\big(\!t,Q_{t^{-}}\!\big)\Big)\mathrm{d}\tilde{N}_{t}^{\Delta,i,j}\!\bigg),
\end{align}
where
\begin{align*}
\mathcal{H}_{t} := \mathcal{H}\Big(Q_t,\mathcal{V}^{+}(t,q),\mathcal{V}^{-}(t,q),v(t,Q_t)\Big)=\sup_{Z\in\mathcal{Z}}h_{t}^{Z}.
\end{align*}
By the fact that $v$ is bounded and $\big(K_{t}^{Z}\big)_{t\in[0,T]}$ is of class $(D)$, the process $\Big(v(t,Q_{t})K_{t}^{Z}\Big)_{t\in[0,T]}$ is a $\mathbb{P}^{\Delta(Z,Q)}-$supermartingale of class $(D)$ and the local martingale term in \eqref{Intermediary step uniform integrability verification} is a true martingale. Hence
\begin{align*}
&v(0,Q_{0})=\mathbb{E}^{\Delta(Z,Q)}\bigg[v(T,Q_{T})K_{T}^{Z}+\int_{0}^{T}K_{t}^{Z}(\mathcal{H}_{t}-h_{t})\mathrm{d}t\bigg]\geq \mathbb{E}^{\Delta(Z,Q)}\big[v(T,Q_{T})K_{T}^{Z}\big]=\mathbb{E}^{\Delta(Z,Q)}\big[-K_{T}^{Z}\big],
\end{align*}
by the boundary condition $v(T,\cdot)=-1$. By arbitrariness of $Z\in \mathcal{Z}$, this provides the inequality 
\begin{align*}
v(0,Q_{0})\geq \sup_{Z\in\mathcal{Z}}\mathbb{E}^{\Delta(Z,Q)}[-K_{T}^{Z}]=v_{0}^{E}. 
\end{align*}
On the other hand, consider the maximiser $Z^{\star}\big(t,Q_{t^{-}}\big)$ in \eqref{Lemma maximizers HJB}. As $(Z^{\star})_{t\in [0,T]}$ is a bounded process, integrability conditions \eqref{Integrability market maker} and \eqref{Integrability principal} are satisfied. Hence, $Z^{\star}\in \mathcal{Z}$. By definition,
\begin{align*}
h^{Z^{\star}}-\mathcal{H}=0,
\end{align*}
thus leading to
\begin{align*}
v(0,Q_{0})=\mathbb{E}^{\Delta(Z^{\star},Q)}\big[-K_{T}^{Z^{\star}}\big].
\end{align*}
Hence, $v(0,Q_{0})=v_{0}^{E}$, with optimal control $Z^{\star}$.

\subsection{First--best exchange problem} \label{Section principal First Best}

In this section, we consider the case of the first best problem. In this particular setting, the principal can control both the spreads quoted by the agents and the contracts given to them. Hence, the exchange manages all the control processes. The goal of this section is to show that the first best problem differs from the second best that we solved throughout this paper. We first introduce the Lagrange multipliers $\lambda:=(\lambda_{i})_{i=1,\dots,N}$ associated to the participation constraints of the agents. For any finite dimensional vector space $E$, with given norm $\|\cdot\|_{E}$, we also introduce the so--called Morse--Transue space on a given probability space $(\Omega,\mathcal{F},\mathbb{P})$, defined by
\begin{align*}
M^{\phi}(E):=\Big\{\xi:\Omega\longrightarrow E \text{ measurable, } \mathbb{E}[\phi(a\xi)]<+\infty, \text{ for any } a\geq 0 \Big\},
\end{align*}
where $\phi:E\longrightarrow \mathbb{R}$ is the Young function, namely $\phi(x)=\exp(\|x\|_E)-1$. Then, if $M^\phi(E)$ is endowed with the norm $\|\xi\|_\phi := \inf \{k>0,\mathbb{E}[\phi(\xi/k)]\leq 1\}$ it is a (non-reflexive) Banach space. 

\medskip
The principal's problem can be reformulated as
\begin{align}\label{Problem principal first best}
V_{0}^{\rm FB}:=\inf_{\lambda>0}\sup_{(\xi,\delta) \in \mathcal{C}\times \mathcal{A}}\mathbb{E}^{\delta}\bigg[-\mathrm{e}^{-\eta(c(N_{T}^{a}+N_{T}^{b})-\xi\cdot\mathrm{1}_N)}-\sum_{i=1}^{N}\lambda_{i}\mathrm{e}^{-\gamma_{i}(\xi^{i}+X_{T}^{i}+Q_{T}^{i}S_{T})}-\lambda_{i}R_{i}\bigg].
\end{align}
If $\lambda$ and $\delta$ are fixed, we start with the maximisation with respect to $\xi$. We introduce the following map $\Lambda^{\delta}:M^{\phi}(\mathbb{R}^N)\rightarrow\mathbb{R}$ defined as
\begin{align*}
\Lambda^{\delta}(\xi):=\mathbb{E}^{\delta}\bigg[-\mathrm{e}^{-\eta(c(N_{T}^{a}+N_{T}^{b})-\xi\cdot\mathrm{1}_N)}-\sum_{i=1}^{N}\lambda_{i}\mathrm{e}^{-\gamma_{i}(\xi^{i}+X_{T}^{i}+Q_{T}^{i}S_{T})}-\lambda_{i}R_{i}\bigg]
\end{align*}
The Lagrange multipliers $\lambda_i$ being strictly positive, using the boundedness of the control process $\delta\in\mathcal{A}$, the map $\Lambda^{\delta}$ is continuous, strictly concave, and G\^ateaux differentiable with, for $h\in M^{\phi}(\mathbb{R}^N)$
\begin{align*}
D\Lambda^{\delta}(\xi)[h]=\mathbb{E}^{\delta}\bigg[-\eta h\cdot \mathrm{1}_{N} \mathrm{e}^{-\eta(c(N_{T}^{a}+N_{T}^{b})-\xi\cdot\mathrm{1}_N)}+\sum_{i=1}^{N}\gamma_i\lambda_{i}h^i\mathrm{e}^{-\gamma_{i}(\xi^{i}+X_{T}^{i}+Q_{T}^{i}S_{T})}\bigg]
\end{align*}

For any $\delta\in\mathcal{A}$, we define
\begin{align*}
&\xi^{\star,i}(\delta):=\frac{\eta}{\gamma_{i}}\big(c(N_{T}^{a}+N_{T}^{b})-\xi^{\star}\cdot \mathbf{1}_N\big)+\frac{1}{\gamma_{i}}\log\bigg(\frac{\lambda_{i}\gamma_{i}}{\eta}\bigg)-(X_{T}^{i}+Q_{T}^{i}S_{T}), \\
&\xi^{\star}(\delta)\cdot\mathbf{1}_N=\frac{1}{1+\eta\Gamma}\bigg(\eta c\Gamma (N_{T}^{a}+N_{T}^{b})+\sum_{i=1}^{N}\frac{1}{\gamma_{i}}\log\bigg(\frac{\lambda_{i}\gamma_{i}}{\eta}\bigg)-(X_{T}^{i}+Q_{T}^{i}S_{T})\bigg)
\end{align*}
with $\Gamma:=\sum_{i=1}^{N}\frac{1}{\gamma_{i}}$. For any $h\in M^{\phi}(\mathbb{R}^N)$, first order condition gives $D\Lambda^{\delta}(\xi)[h]=0$. Computations show that $(\xi^{\star,i})_{i=1,\dots,N}$ achieve the maximum of $\Lambda^{\delta}(\xi)$, hence is optimal for \eqref{Problem principal first best}. Then, substituting these expressions in the main problem gives
\begin{align}
V_{0}^{\rm FB}=(1+\eta\Gamma)\inf_{\lambda_>0}\prod_{i=1}^{N}\bigg(\frac{\lambda_{i}\gamma_{i}}{\eta}\bigg)^{\eta\gamma_{i}^{-1}(1+\eta\Gamma)^{-1}}\widetilde{V}_{0}- \sum_{i=1}^{N}\lambda_{i}R_{i} , 
\end{align}
where 
\begin{align}
\widetilde{V}_{0}:=\sup_{\delta \in \mathcal{A}}\mathbb{E}^{\delta}\bigg[-\exp\bigg(-\frac{\eta}{1+\eta\Gamma}\bigg(\sum_{i=1}^{N}(X_{T}^{i}+Q_{T}^{i}S_{T})+c(N_{T}^{a}+N_{T}^{b})\bigg)\bigg)\bigg].
\end{align}
This is a stochastic control problem, see \cite[Section A.7]{el2018optimal} for details, whose HJB equation is given by
\begin{align}\label{HJB equation First Best}
\begin{cases}
     \displaystyle   \partial_{t}v(t,q) +v(t,q)\frac{1}{2}\sigma^{2}\tilde{\Gamma}^{2}\|q\|^2+ \mathcal{H}^{\text{FB}}\big(q,\mathcal{V}^{+}(t,q),\mathcal{V}^{-}(t,q),v(t,q)\big)=0, \; (t,q)\in[0,T)\times\mathcal{Q}^{N},\\
   \displaystyle     v(T,q)=-1,\; q\in\Qc^N,
    \end{cases}
\end{align}
where $\tilde{\Gamma}:=\frac{\eta}{1+\eta\Gamma}$, and
\begin{align*}
\mathcal{H}^{\text{FB}}\big(q,\mathcal{V}^{+}(t,q),\mathcal{V}^{-}(t,q),v(t,q)\big):=\mathcal{H}^{\text{FB},b}\big(q,\mathcal{V}^{+}(t,q),v(t,q)\big)+\mathcal{H}^{\text{FB},a}\big(q,\mathcal{V}^{-}(t,q),v(t,q)\big),
\end{align*}
with, for any $(p,v,j)\in\R^N\times\R\times\{a,b\}$
\begin{align*}
& \mathcal{H}^{\text{FB},j}\big(q,p,v\big):= \sup_{\delta^{j} \in \mathcal{A}}\sum_{i=1}^{N}\lambda^{i,j}(\delta^{j},q)\bigg(\exp\bigg(-\tilde{\Gamma}\bigg(\sum_{i=1}^{N}\delta^{i,j}\mathbf{1}_{\{\delta^{i,j}=\underline{\delta}^{j}\}}+\sum_{\ell=1}^{K}\omega_{\ell}\delta^{i,j}\mathbf{1}_{\{\delta^{i,j}\in K_{\ell}\}}\bigg)\bigg)p^i-v\bigg).
\end{align*}
We are in a framework similar to \cite{el2018optimal,gueant2013dealing}. First order condition gives for $j\in\{a,b\}$
\begin{align*}
& \sum_{i=1}^{N}\delta^{\star,j,i}\mathbf{1}_{\{\delta^{\star,j,i}=\underline{\delta}^{\star,j,i}\}}\!\!+\!\!\sum_{\ell=1}^{K}\omega_{\ell}\delta^{\star,j,i}\mathbf{1}_{\{\delta^{\star,j,i}\in K_{\ell}\}}\!=\!\mathcal{P}^{j}(t,q)\!:=\!\frac{1}{\tilde{\Gamma}}\Bigg(\!\log\!\bigg(\!1\!+\!\frac{\tilde{\Gamma}\sigma}{k\underline{\omega}}\bigg)\!\!+\!\!\log\bigg(\frac{\sum_{i\in \mathcal{G}}v(t,q\ominus_i \phi(j))}{v(t,q)}\bigg) \Bigg).
\end{align*}
Such conditions are satisfied with the following optimal bid--ask policy, for $j\in\{a,b\}$
\begin{align*}
\delta^{\star,j,i}(t,q) := 
    \begin{cases}
      \displaystyle  (-\delta_{\infty})\vee \frac{1}{\omega_{\ell}}\mathcal{P}^{j}(t,q)\wedge \delta_{\infty},\; {\rm if}\; \frac{1}{\omega_{\ell}}\mathcal{P}^{a}(t,q)\in K_{\ell}, \mbox{ for } \ell\in\{1,\dots,K\}, \\
  \displaystyle      (-\delta_{\infty})\vee \mathcal{P}^{j}(t,q) \wedge \delta_{\infty},\; \mbox{otherwise}.
    \end{cases}
\end{align*}
Finally, computations show that the Hessian associated to the supremum in $\delta^{\star,j}$ is symmetric definite negative,
hence, $\delta^\star$ is a local maximum.
\begin{theorem}
There exists a unique negative bounded solution to the {\rm PDE} 
\begin{align*}
\begin{cases}
        \displaystyle \partial_{t}v(t,q)\!+\!v(t,q)\bigg(\frac{\sigma^{2}}{2}\tilde{\Gamma}^{2}\|q\|^{2}\!-\! \tilde{C}^{\text{FB}}\bigg(\bigg(\frac{v(t,q)}{\sum_{i\in \mathcal{G}}v(t,q\oplus_i 1)}\bigg)^{\frac{k\varpi}{\sigma\tilde{\Gamma}}}\!\!+\!\bigg(\frac{v(t,q)}{\sum_{i\in \mathcal{G}}v(t,q\ominus_i 1)}\bigg)^{\frac{k\varpi}{\sigma\tilde{\Gamma}}}\bigg)\bigg)\!\!=0,\; ,  \\
        v(T,q)=-1,\; 
\end{cases}
\end{align*}
with $(t,q)\!\in \![0,T)\!\times\! \mathcal{Q}^N$, $\tilde{C}^{\text{FB}}:=A\text{ exp}\big(-\frac{k\varpi}{\sigma\tilde{\Gamma}}\log\big(1+\frac{\tilde{\Gamma}\sigma}{k\varpi}\big) \big)\frac{\tilde{\Gamma}\sigma}{k\varpi+\tilde{\Gamma}\sigma}$. Moreover, this solution coincides with the value function of the exchange for the problem \eqref{Problem principal first best}.
\end{theorem}
The proof is omitted as it relies on the same basis arguments as Lemma \ref{Lemma Cauchy Lipschitz} and Theorem \ref{Theorem Verification}.

\medskip
We now prove that the PDE satisfied by the value function of the First Best problem is different from the one verified in the second best case \eqref{PDE non-linear before change of variable}. Taking the special case $\gamma_{i}:=\gamma$, i.e the case of market makers with same risk aversion, for $i\in\{1,\dots,N\}$, the PDE boils boils down to
\begin{align*}
\begin{cases}
        \displaystyle \partial_{t}v(t,q)+v(t,q)C^{\text{FB}}(q)-v(t,q)\tilde{C}^{\text{FB}}\sum_{j\in\{a,b\}}\bigg(\frac{v(t,q)}{\sum_{i=1}^{N}\mathbf{1}_{\{\phi(j)q^{i}>-\overline{q}\}}v(t,q^{i}-\phi(j))}\bigg)^{\frac{k\varpi}{\sigma\eta}}=0,\; ,  \\
        v(T,q)=-1,\; ,
\end{cases}
\end{align*}
where $(t,q)\!\in \![0,T)\!\times\! \mathcal{Q}^N$ and
\begin{align*}
C^{\text{FB}}(q):=\frac{1}{2}\sigma^{2}\tilde{\Gamma}^{2}\|q\|^2.
\end{align*}

By noting that for all $q\in \mathcal{Q}^{N},$  $C^{\text{FB}}(q)\neq C^{S}(q)$ and $C \neq  \tilde{C}^{\text{FB}}$, we see that the value function of the exchange in the first best case does not coincide with the value function in the second best model.

\end{appendix}
\bibliography{bibliographyDylan-1.bib}

\end{document}